\documentclass[11pt,a4paper]{article}
\usepackage[tbtags]{amsmath}    
\usepackage{amsfonts,amssymb,amsthm,euscript,makeidx,pifont,boxedminipage,multicol,fullpage}
\usepackage[T1]{fontenc}

\usepackage{bbm}

\newtheorem{theorem}{Theorem}[section]
\newtheorem{lemma}[theorem]{Lemma}
\newtheorem{proposition}[theorem]{Proposition}
\newtheorem{corollary}[theorem]{Corollary}
\newtheorem{remark}[theorem]{Remark}

\newtheorem{definition}[theorem]{Definition}
\newtheorem{assumption}[]{Assumption}
\newtheorem{assumption-H}[]{Assumption-H}
\newtheorem{notations}[]{Notations}
\newtheorem{Keywords}[theorem]{Keywords}

\usepackage{color}

\newcommand\R{\mathbb{R}}

\usepackage{graphicx}

\usepackage{amsfonts,amssymb,amsthm,euscript,makeidx,pifont,boxedminipage,multicol,fullpage}
\usepackage[T1]{fontenc}
\usepackage{amsmath}

\usepackage{bbm}
\numberwithin{equation}{section} 

\begin{document}

\bibliographystyle{plain}

\title{Non-linear filtering and optimal investment under partial information for stochastic volatility models}

\date{\today}

\author{Dalia Ibrahim; Fr\'{e}d\'{e}ric Abergel\thanks{Ecole Centrale Paris, Laboratoire de Math\'ematiques Appliqu\'es aux Syst\`emes, Grande Voie des Vignes, 92290 Ch\^{a}tenay Malabry, France};}

\maketitle

\begin{abstract}
This paper studies the question of filtering and maximizing terminal wealth from expected utility in a partially information stochastic volatility models. The special features is that the only information available to the investor is the one generated by the asset prices, and the unobservable processes will be modeled by a stochastic differential equations. Using the change of measure techniques, the partial observation  context can be transformed into a full information context such that coefficients depend only on past history of observed prices (filters processes). Adapting the stochastic non-linear filtering, we show that under some assumptions on the model coefficients, the estimation of the filters depend on a priori models for the trend and the stochastic volatility. Moreover, these filters  satisfy a stochastic partial differential equations named "Kushner-Stratonovich equations". Using the martingale duality approach in this partially observed incomplete model, we can characterize the value function and the optimal portfolio. The main result here is that the dual value function associated to the martingale approach can be expressed, via the dynamic programming approach, in terms of the solution to a semilinear partial differential equation. We illustrate our results with some examples of stochastic volatility models popular in the financial literature.  
\end{abstract}

\begin{Keywords}
Partial information, stochastic volatility, utility maximization, martingale duality method, non-linear filtering, Kushner-Stratonovich equations, semilinear partial differential equation.  
\end{Keywords}

\section{Introduction}
 
 \vspace{3mm}

 The basic problem of mathematical finance is the problem of an economic agent who invests in a financial market so as to maximize the expected utility of his terminal wealth. In the framework of continuous time model, the utility maximization problem has been studied for the first time by Merton (1971) in a Black-Scholes environment (full information) via the Hamilton-Jaccobi-Bellman equation and dynamic programming. As in financial market models, we do not have in general a complete knowledge of all the parameters, which may be driven by unobserved random factors. So, we are in the situation of the utility maximization problem with partial observation, which has been studied extensively in the literature by Detemple \cite{Detemple}, Dothan and Feldman \cite{Feldman}, Lakner \cite{Lakner1}, \cite{Lakner2}, etc. There are many generalizations of Merton's setting. The natural generalizations was to model the volatility by a stochastic process.

In this paper, we consider a financial market where the price process of risky asset follows a stochastic volatility model and we require that investors observe just the stock price. So we are in the framework of partially observed incomplete market, where our aim is to solve the utility maximization problem in this context.

In order to solve this problem with partial observation, the common way is to use the stochastic non-linear filtering and change of measure techniques, so as the partial observation context can be transformed into a full information context. Then it is possible to solve this problem either with the martingale approach or via dynamic programming approach. Models with incomplete information have been investigated by Dothan and Feldman \cite{Feldman} using dynamic programming methods in a linear Gaussian filtering, Lakner \cite{Lakner1}, \cite{Lakner2} has solved the partial optimization problem via martingale approach and worked out the special case of the linear Gaussian filtering. Pham and Quenez \cite{Pham-Quenez} treated the case of partial information stochastic volatility model where they have combined stochastic filtering techniques and a martingale duality approach to characterize the value function and the optimal portfolio of the utility maximization problem. They have studied two cases: the case where  the risks of the model are assumed to be independent Gaussian processes  and the Bayesian case studied by Karatzas-Zhao \cite{Karatzas-Zhao}. 

\smallskip
In this paper, we are in the same framework studied by Pham and Quenez \cite{Pham-Quenez}, but here we assume that the unobservable processes are modeled by a stochastic differential equations. More precisely, the unobservable drift of the stock and that of the stochastic volatility are modeled by stochastic differential equations. The main result in this case, is that the filters estimate of the risks depend on a priori models for the trend and the stochastic volatility. There are two reasons for this result: Firstly, we need to choose the models of the trend and the stochastic volatility such that the risks dynamics can be described only in terms of them. Secondly, we need to  choose these models such that the coefficients of the risks dynamics satisfy some regularity assumptions, like globally Lipshitz conditions and some finite order moment will be imposed. We show that the filters estimate of the risks satisfy a stochastic partial differential equations named "Kushner-Stratonovich equations". But these equations are valued in infinite dimensional space and cannot be solved explicitly, so numerical approximitions can be used to resolve them. Also, we study the case of finite dimensional filters like Kalman-Bucy filter. We illustrate our results with several popular examples of stochastic volatility models. 

\smallskip
After replacing the original partial information problem by a full information one which depends only on the past history of observed prices, it is then possible to use the classical theory for stochastic control problem. Here we will be interested by the martingale approach to solve our utility optimization problem.  As the reduced market in incomplete, we complement the martingale approach by using the theory of stochastic control to solve the related dual optimization problem. In \cite{Pham-Quenez}, they have also used the martingale approach, but they have studied the case where the dual optimizer vanishes. The main result in this paper is that the solution of the related dual problem can be expressed in terms of the solution to a semilinear partial differential equation which depends also on the filters and the stochastic volatility.  

\smallskip
The paper is organized as follows: In section $2$, we describe the model and formulate the optmization problem. In section $3$, we use the non-linear filtering techniques and the change of measure techniques in order to transform the partial observation context into a full information context such that coefficients depend only on past history of observed prices (filters processes). In section $4$, we show that the filters estimations depend on a priori models for the trend and the stochastic volatility. We illustrate our results with examples of stochastic volatility models popular in the financial literature. Finally, in section $5$, we use the martingale duality approach for the utility maximization problem.  We show that the dual value function and the dual optimizer can be expressed in terms of the solution to a semilinear partial differential equation. By consequence, the primal vale function and the optimal portfolio depend also on this solution. The special cases of power and logarithmic utility functions are studied and we illustrate our results by an examples of stochastic volatility models for which we can give a closed form to the semilinear  equation.

\section{Formulation of the problem}
Let $\left(\Omega,\mathcal{F}, \mathbb{P}\right)$ be a complete probability space equipped with a filtration $\mathbb{F}=\{\mathcal{F}_{t},0 \leq t \leq T \}$ satisfying the usual conditions, where $T>0$ is a fixed time horizon. The financial market consists of one risky asset and a bank account (bound). The price of the bound is assumed for simplicity to be $1$ over the entire continuous time-horizon $[0,T]$ and the risky asset has dynamics:
\begin{align}
\label{model-S}
&\frac{dS_{t}}{S_{t}}= \mu_{t} dt + g(V_{t}) dW^{1}_{t},\\
\label{model-V}
& dV_{t}=f(\beta_{t},V_{t}) dt + k(V_{t})(\rho dW^{1}_{t}+\sqrt{1-\rho^{2}}dW^{2}_{t}),\\
& d\mu_{t}= \zeta(\mu_{t}) dt + \vartheta(\mu_{t}) dW^{3}_{t}.
\label{model-mu}
\end{align}

\noindent The processes $W^{1}$ and $W^{2}$ are two independents Brownian motions defined on $\left(\Omega,\mathcal{F}, \mathbb{P}\right)$ and $-1\leq \rho\leq 1$ is the correlation coefficient. $W^{3}$ is a standard Brownian motion independent of $W^{1}$ and $W^{2}$. The drift $\mu=\{\mu_{t}, 0\leq t \leq T\}$ is not observable and follows a Gaussian process. The process $\beta_{t}$ can be taken as a function in terms of $\mu_{t}$ or another unobservable process, which also has a stochastic differential equation. 

\vspace{2mm}
\noindent We assume that the functions $g$, $f$, $k$, $\zeta$ and $\vartheta$ ensure existence and uniqueness for solutions to the above stochastic differential equations. A Lipschitz conditions are sufficient, but we do not impose these on the parameters at this stage, as we do not wish to exclude some well-known stochastic volatility models from the outset. Also, we can assume that the drift $\mu_{t}$ can be replaced by $\mu_{t}~g(V_{t})$, that is we have a factor model. 

\smallskip 
\noindent Moreover, we assume that $g(x), k(x)>0$ and the solution of (\ref{model-V}) does not explode, that is, the solution does not touch $0$ or $\infty$ in finite time. The last condition can be verified form Feller's test for explosions given in \cite[p.348]{Shreve}. 

\vspace{6mm}
\noindent In the sequel, we denote by $\mathbb{F}^{S}=\{F_{t}^{S}, 0\leq t \leq T\}$ (resp. $\mathbb{F}^{V}=\{F_{t}^{V}, 0 \leq t \leq T\}$) the filtration generated by the price process $S$ (resp. by the stochastic volatility $V$). Also we denote by $\mathbb{G}=\{\mathcal{G}_{t}, 0 \leq t \leq T\}$ the natural $\mathbb{P}$-augmentation of the market filtration generated by the price process $S$.

\subsection{\textbf{The optimization problem}}

Let $\pi_{t}$ be the fraction of the wealth that the trader decides to invest in the risky asset at time $t$, and  $1-\pi_{t}$ is the fraction of wealth invested in the bound. We assume that the trading strategy is self-financing, then the wealth process corresponding to a portfolio $\pi$ is defined by $R_{0}^{\pi}=x$ and satisfies the following $SDE$:
\begin{align*}
dR_{t}^{\pi}=R_{t}^{\pi}\left(\pi_{t}\mu_{t}dt +\pi_{t} g(V_{t}) dW^{1}_{t}\right).
\end{align*} 
 A function $U: \mathbb{R} \rightarrow \mathbb{R}$ is called a utility function if it is strictly increasing, strictly concave of class $C^{2}$. We assume that the investor wants to maximize the expected utility of his terminal wealth. The optimization problem thus
reads as
\begin{equation}\label{Optim-prob-1} 
J(x)=\displaystyle\sup_{\pi\in \mathcal{A}} \mathbb{E}[U(R_{T}^{\pi})],~~~~ x>0,
\end{equation}

\noindent where $\mathcal{A}$ denotes the set of the admissible controls $(\pi_{t}, 0\leq t \leq T)$
which are $\mathbb{F}^{S}$-adapted, and satisfies the integrability condition:
\begin{equation}\label{integrability-cond}
\displaystyle\int_{t}^{T} g^{2}(V_{s})\pi_{s}^{2}ds <\infty~~~~~~\mathbb{P}-a.s.
\end{equation}

We are in a context when an investor wants to maximize the expected utility from terminal wealth,  where the only information available to the investor is the one generated by the
asset prices, therefore leading to a utility maximization problem in partially observed incomplete model. In order to solve it, we aim to reduce it to a maximization problem with full information. For that, it becomes important to exploit
all the information coming from the market itself in order to continuously
update the knowledge of the not fully known quantities and this is where stochastic filtering becomes useful.

\section{Reduction to a full observation context}\label{Reduction}
\noindent Let us consider the following processes:
\begin{align}
\label{mu-conditional}
&\tilde{\mu}_{t}:=\frac{\mu_{t}}{g(V_{t})},\\
&\tilde{\beta}_{t}:=\left(\sqrt{1-\rho^{2}}k(V_{t})\right)^{-1}\left(f(\beta_{t},V_{t})-\rho k(V_{t})\tilde{\mu}_{t}\right),
\label{beta-conditional}
\end{align} 
we assume that they verify the integrability condition:
\[
\displaystyle\int_{0}^{T}|\tilde{\mu}_{t}|^{2}+|\tilde{\beta}_{t}|^{2} dt <\infty ~~~~ \rm{a.s}.
\]

\noindent Here $\tilde{\mu}_{t}$ and $\tilde{\beta}_{t}$ are the unobservable processes that account for the market price of risk. The first is related to the asset's Brownian component. The second to the stochastic volatility's Brownian motion.  


\vspace{3mm}
\noindent Also we introduce the following process:
\begin{equation}\label{Hmar}
L_{t}=1-\displaystyle \int_{0}^{t}L_{s}\left[\tilde{\mu}_{s}dW^{1}_{s}+ \tilde{\beta}_{s}dW^{2}_{s}\right].
\end{equation}

\noindent We shall make the usual standing assumption of filtering theory. 

\vspace{2mm}

\begin{assumption}\label{Martingale}
 The process $L$ is a martingale, that is, $\mathbb{E}[L_{T}]=1.$ 
 \end{assumption}

\vspace{2mm}
\noindent Under this assumption, we can now define a new probability measure $\tilde{\mathbb{P}}$ equivalent to $\mathbb{P}$ on $\left(\Omega,\mathbb{F}\right)$ characterized by:  

\begin{equation}\label{equivalent-pro}
\frac{d\tilde{\mathbb{P}}}{d\mathbb{P}}|\mathcal{F}_{t}=L_{t},~~~~~0\leq t \leq T.
\end{equation}

\noindent Then Girsanov's transformation ensures that
\begin{align}
\label{first-change}
&\tilde{W}^{1}_{t}=W^{1}_{t}+\displaystyle\int_{0}^{t}\tilde{\mu_{s}}ds~~\mbox{is a $(\tilde{\mathbb{P}},\mathbb{F})$-Brownian motion},\\
& \tilde{W}^{2}_{t}=W^{2}_{t}+\displaystyle\int_{0}^{t}\tilde{\beta_{s}}ds~~ \mbox{is a $(\tilde{\mathbb{P}},\mathbb{F})$-Brownian motion}.\\
\label{second-change}
\end{align} 
\noindent Also, we have that $(\tilde{\mu}_{t},\tilde{\beta}_{t})$ is independent of the Brownian motion $\left(\tilde{W}^{1}_{t}, \tilde{W}^{2}_{t}\right)$.

\noindent Therefore, the dynamics of $(S,V)$ under $\tilde{\mathbb{P}}$ become:
\begin{align}
\label{S-girsanov}
&\frac{dS_{t}}{S_{t}}= g(V_{t}) d\tilde{W}^{1}_{t},\\
&dV_{t}=\rho~k(V_{t}) d\tilde{W}_{t}^{1}+ \sqrt{1-\rho^{2}}~k(V_{t}) d\tilde{W}_{t}^{2}.
\label{V-girsanov}
\end{align}

\vspace{2mm}

\noindent We now state a lemma which will highly relevant in the following. The proof of this lemma is similar to lemma $3.1$ in Pham and Quenez \cite{Pham-Quenez}. 

\begin{lemma}\label{aug-filt}
Under assumption \ref{Martingale}, the filtration $\mathbb{G}$ is the augmented filtration of $(\tilde{W}^{1},\tilde{W}^{2})$. 
\end{lemma}

\begin{proof}
The sketch of the proof is summarized by two steps: 

\noindent Firstly, we show that the filtration $\mathbb{G}$ is equal to the enlarged progressive filtration $\mathbb{F}^{S}\vee \mathbb{F}^{V}$. The first inclusion is obvious and the other inclusion $\mathbb{F}^{S}\vee \mathbb{F}^{V} \subset \mathbb{G}$ is deduced from the fact that $V_{t}$ can be estimated from the quadratic variation of $log(S_{t})$. Secondly, from (\ref{S-girsanov}), (\ref{V-girsanov}) and the fact that $g(x),k(x)>0$, we have that $\mathbb{F}^{\tilde{W}^{1}}\bigvee \mathbb{F}^{\tilde{W}^{2}}$ the filtration generated by $(\tilde{W}^{1},\tilde{W}^{2})$

\end{proof}

\vspace{5mm}

\noindent We now make the following assumption on the risk processes $\left(\tilde{\mu},\tilde{\beta}\right)$.
 \begin{equation}\label{risk-assumptions}
 \forall t\in[0,T],~~~~ \mathbb{E}|\tilde{\mu}_{t}|+\mathbb{E}|\tilde{\beta}_{t}| < \infty
 \end{equation}
\noindent Under this assumption, we can introduce the conditional law of $\left(\tilde{\mu},\tilde{\beta}\right)$:
\begin{align}
\label{filter-mu}
&\overline{\mu}_{t}:=\mathbb{E}[\tilde{\mu}_{t}|\mathcal{G}_{t}],\\
\label{filter-beta}
&\overline{\beta}_{t}:=\mathbb{E}[\tilde{\beta}_{t}|\mathcal{G}_{t}].
\end{align}

\noindent Let us denote by $H$ the $(\tilde{\mathbb{P}},\mathbb{F})$ martingale defined as $H_{t}=\dfrac{1}{L_{t}}$. Now, we aim to construct the restriction of $\mathbb{P}$  equivalent to $\tilde{\mathbb{P}}$ on $(\Omega,\mathbb{G})$. First, let us consider the conditional version of Baye's formula: for any $\mathbb{P}$ integrable random variable $X$ ($X\in L^{1}(\mathbb{P})$), we have:
\begin{equation}\label{Bayes-restriction}
\mathbb{E}\left[X|\mathcal{G}_{t}\right]=\frac{\tilde{\mathbb{E}}\left[X H_{t}|\mathcal{G}_{t}\right]}{\tilde{\mathbb{E}}\left[H_{t}|\mathcal{G}_{t}\right]}.
\end{equation} 

\noindent Then by taking $X=L_{t}$, we get:
\begin{align}\label{girs-restriction}
\tilde{L}_{t}:=\mathbb{E}\left[L_{t}|\mathcal{G}_{t}\right]=\frac{1}{\tilde{\mathbb{E}}[H_{t}|\mathcal{G}_{t}]}.
\end{align}

\noindent Therefore, from (\ref{equivalent-pro}) (\ref{girs-restriction}), we have the following restriction to $\mathbb{G}$:
\begin{equation*}
\frac{d\tilde{\mathbb{P}}}{d\mathbb{P}}|\mathcal{G}_{t}=\tilde{L}_{t}.
\end{equation*}

\noindent Finally, from Bain and Crisan (proposition $2.30$) and Pardoux (proposition $2.2.7$), we have the following result:

\begin{proposition}\label{innovation-proc}
The following processes $\overline{W}^{1}$ and $\overline{W}^{2}$ are independent $\left(\mathbb{P},\mathbb{G}\right)$-Brownian motions.\begin{align*}
&\overline{W}^{1}_{t}=
W^{1}_{t}+\displaystyle\int_{0}^{t}\left(\tilde{\mu}_{s}-
\overline{\mu}_{s}\right)ds:=\tilde{W}^{1}_{t}-\displaystyle\int_{0}^{t}\overline{\mu}_{s}ds,\\
&\overline{W}^{2}_{t}=W^{2}_{t}+\displaystyle\int_{0}^{t}\left(\tilde{\beta}_{s}
-\overline{\beta}_{s}\right)ds:=\tilde{W}^{2}_{t}-\displaystyle\int_{0}^{t}\overline{\beta}_{s}ds.
\end{align*}
\end{proposition}
\noindent These processes are called the innovation processes in filtering theory.  They include the distances between the true values of $\tilde{\mu}$ and $\tilde{\beta}$ and their estimates:

\vspace{7mm}
\noindent Then, by means of the innovation processes, we can describe the dynamics of $(S,V,R)$ within a framework of full observation model:
$$
(Q) =
\left\{
\begin{array}{lr}
\frac{dS_{t}}{S_{t}}=g(V_{t})\overline{\mu}_{t}dt + g(V_{t})d\overline{W}^{1}_{t},\\

\vspace{3mm}

dV_{t}=\left(\rho ~k(V_{t})\overline{\mu}_{t}+ \sqrt{1-\rho^{2}} ~k(V_{t})\overline{\beta}_{t} \right)dt + \rho k(V_{t}) d \overline{W}^{1}_{t} + \sqrt{1-\rho^{2}} k(V_{t})d \overline{W}^{2}_{t},\\

\vspace{2mm}

dR_{t}^{\pi} =  R_{t}^{\pi}\pi_{t}\left(g(V_{t})~\overline{\mu}_{t}dt + g(V_{t})d \overline{W}^{1}_{t}\right).
\end{array}
\right.
$$

\vspace{3mm}

\section{Filtering}\label{Filtering}
We have showed that conditioning arguments can be used to replace the initial partial information problem by a full information problem one which depends only on the past history of observed prices. But the reduction procedure involves the filters estimate  $\overline{\mu}_{t}$ and $\overline{\beta}_{t}$.

\vspace{3mm}
Our filtering problem can be summarized as follows: 
From lemma \ref{aug-filt}, we have $\mathbb{G}=\mathbb{F}^{\tilde{W}^{1}}\vee \mathbb{F}^{\tilde{W}^{2}}$. Then  the vector $(\tilde{W}^{1},\tilde{W}^{2})$ corresponds to the observation  process. On the other hand, our signal process is given by $(\tilde{\mu_{t}},\tilde{\beta}_{t})$. So the filtering problem is to characterize the conditional distribution of $(\tilde{\mu_{t}},\tilde{\beta}_{t})$, given the observation data $\mathbb{G}=\mathbb{F}^{\tilde{W}^{1}\bigvee \tilde{W}^{2}}$.

\smallskip\smallskip
We show in this section how the filters estimate depend on the models of the drift and the stochastic volatility. Using the non-linear filtering theory (presenting in appendix), we can deduce that the filters estimate satisfy some stochastic partial differential equations, called "Kushner-Stratonovich equations". Generally these equations are infinite-dimensional and thus very hard to solve them explicitly. So, in order to simplify the situation and in order to obtain a closed form for the optimal portfolio, we will be interested by some cases of models, when we can deduce a finite dimensional filters.

\subsection{General Case:}
\vspace{3mm}
\noindent Let us assume that 
the processes $\tilde{\mu}_{t}$ and $\tilde{\beta}_{t}$ are solutions of the following stochastic differential equations: 

   \begin{eqnarray}\label{general-filtering}
d\left(
\begin{array}{c}
\tilde{\mu}_{t} \\
\tilde{\beta}_{t}
\end{array} \right)= \left(
\begin{array}{c}
a \\
\overline{a}
\end{array} \right)dt+ \left(
\begin{array}{cc}
g_{1} & g_{2}\\
\overline{g}_{1} & \overline{g}_{2}
\end{array} \right)d\left(
\begin{array}{c}
W^{3}_{t} \\
W^{4}_{t}
\end{array} \right)+
\left(
\begin{array}{cc}
b_{1} & b_{2}\\
\overline{b}_{1} & \overline{b}_{2}
\end{array} \right)d\left(
\begin{array}{c}
W^{1}_{t} \\
W^{2}_{t}
\end{array} \right)
\end{eqnarray}
 
\noindent where we denote for simplification the functions $a:=a(\tilde{\mu}_{t},\tilde{\beta}_{t}), \overline{a}:=\overline{a}(\tilde{\mu}_{t},\tilde{\beta}_{t})$, ....... $\overline{b}_{2}=\overline{b}_{2}(\tilde{\mu}_{t},\tilde{\beta}_{t})$, and the Brownian motion $(W^{3}_{t},W^{4}_{t})$ is independent of $(W^{1}_{t},W^{2}_{t})$.

\vspace{2mm}
\noindent On the other hand, the dynamics of the observation process $(\tilde{W}^{1},\tilde{W}^{2})$ is given by:
\begin{equation}
\label{observation-process}
d\left(
\begin{array}{c}
\tilde{W}^{1}_{t} \\
\tilde{W}^{2}_{t}
\end{array} \right)=d\left(
\begin{array}{c}
W^{1}_{t} \\
W^{2}_{t}
\end{array} \right)+\left(
\begin{array}{c}
\tilde{\mu}_{t} \\
\tilde{\beta}_{t}
\end{array} \right) dt
\end{equation}

\begin{remark}
To avoid confusion in the sequel, we have:  $\mathbb{G}=\mathbb{F}^{\tilde{W}^{1}\bigvee \tilde{W}^{2}}=\mathbb{F}^{Y}$.
\end{remark}

\begin{notations}\label{notations-1}
Let us denote by: 
\[
X_{t}=\left(
\begin{array}{c}
\tilde{\mu}_{t}\\
\tilde{\beta}_{t}
\end{array} \right),~~~ Y_{t}=\left(
\begin{array}{c}
\tilde{W}_{t}^{1}\\
\tilde{W}^{2}_{t}
\end{array} \right),~~~ A=\left(
\begin{array}{c}
a \\
\overline{a}
\end{array} \right),~~~G=\left(
\begin{array}{cc}
g_{1}& g_{2}\\
\overline{g}_{1} & \overline{g}_{2}
\end{array} \right),~~~ B=\left(
\begin{array}{cc}
b_{1} & b_{2}\\
\overline{b}_{1} & \overline{b}_{2}
\end{array} \right)
\]
\begin{equation}
M_{t}=\left(
\begin{array}{c}
W_{t}^{3} \\
W_{t}^{4}
\end{array} \right)
W_{t}=\left(
\begin{array}{c}
W_{t}^{1} \\
W_{t}^{2}
\end{array} \right),~~~ h=\left(
\begin{array}{c}
h_{1} \\
h_{2}
\end{array} \right)~~~~K=\dfrac{1}{2}(BB^{T}+GG^{T}).\label{notation}\end{equation}
where for $ x=(m,b), h_{1}(x)=m$, $h_{2}(x)=b$ and $T$ denotes the the transposition operator.  
\end{notations}
\vspace{3mm}

\noindent With these notations, the signal-observation processes $(X_{t},Y_{t})$ satisfy (\ref{signal-process}) and (\ref{observation-process}):
\begin{align}
\label{signal-process-section}
&dX_{t}=A(X_{t})dt +G(X_{t})dM_{t}+ B(X_{t})dW_{t}\\
&dY_{t} = dW_{t}+h(X_{t}) dt
\label{observation-process-section}
\end{align}

\smallskip\smallskip
\subsubsection{\textbf{Estimate $\overline{\mu}_{t}$ and $\overline{\beta}_{t}$}}\label{estimate-filters}

Let us now make some assumptions which will be useful to show our results.

\underline{\textbf{Assumptions}}
\begin{itemize}
\item $i)$ The functions $A, G$ and $B$ are globally Lipschitz.
\item $ii)$ $X_{0}$ has finite second moment.
\item $iii)$ $X_{0}$ has finite third moment.
\end{itemize}

\begin{lemma}\label{mart-proof}
Let $(X,Y)$ be the solution of (\ref{signal-process-section}) and (\ref{observation-process-section}) and assume that $h$ has linear growth condition. If assumptions $i)$ and $ii)$ are satisfied, then (\ref{assumption-L}) is satisfied. Moreover, if assumption $iii)$ is satisfied, then (\ref{cond-zakai-1}) is satisfied.
\end{lemma}
\begin{proof}
The proof is given  in \cite{Bensoussan}(see, lemma $4.1.1$ and lemma $4.1.5$).
\end{proof}

\vspace{2mm}
The following results show that we need to introduce an a priori models for the trend and the stochastic volatility in order to describe the dynamics of  $(\tilde{\mu}_{t},\tilde{\beta}_{t})$ as in (\ref{general-filtering}), and therefore we can deduce from proposition \ref{general-Kushner-equation}  the dynamics of the filters estimate $(\overline{\mu}_{t},\overline{\beta}_{t})$ and therefore deduce that of $(\overline{\mu}_{t}, \overline{\beta}_{t})$. More precisely, we show that these estimates depend essentially on the model of the volatility $V_{t}$. We need to choose the dynamics of $V_{t}$ such that the following two steps will be verified.

\begin{itemize}
\item \underline{First step:} Describe the dynamics of $(\tilde{\mu}_{t}, \tilde{\beta}_{t})$ as in (\ref{general-filtering})

We show that this description depend essentially on the model of $V_{t}$. In fact, if we apply It\^{o}'s  formula on $\tilde{\mu}_{t}$ and $\tilde{\beta}_{t}$ in order to describe their dynamics, we have that $V_{t}$ still appear, for that we need to describe $V_{t}$ only in terms of $\tilde{\mu}_{t}$ and $\tilde{\beta}_{t}$ in order to disappear it from their dynamics. This can be done from the definition of the $\tilde{\beta}_{t}$ but taking in account the choice of the variable $\beta_{t}$ or more precisely the choice of $f(\beta_{t},V_{t})$. We will clarify this  with an examples in paragraph $4.1.1$.  

\item \underline{Second step:} Verification of some regularity assumptions

Once we describe the dynamics of  $(\tilde{\mu}_{t}, \tilde{\beta}_{t})$ as in (\ref{general-filtering}), we must check in more that the coefficients of the dynamics verify some regularity assumptions, in order to use the above results of nonlinear filtering theory. 
\end{itemize}

\noindent We present now our result concerning the filtering problem: 

\begin{proposition}\label{dependence-drift}
 We assume that there exists a function $\Upsilon: \mathbb{R}^{2}\to\mathbb{R}$ such that $V_{t}=\Upsilon(\tilde{\mu}_{t},\tilde{\beta}_{t})$. If with this function, the dynamics of $X_{t}=(\tilde{\mu}_{t},\tilde{\beta}_{t})$ can be described as in (\ref{signal-process-section}) and assumptions $i),ii)$ and $iii)$ hold, then the conditional distribution $\alpha_{t}:\mathbb{E}[\phi(X_{t}|\mathbb{F}_{t}^{Y})]$ satisfy the following Kushner-Stratonovich equation:
 \begin{align}
 \nonumber
d\alpha_{t}(\phi)&=\alpha_{t}(A\phi)dt +\left[\alpha_{t}\left(\left(h^{1}+\mathcal{B}^{1}\right)\phi\right)-\alpha_{t}(h^{1})\alpha_{t}(\phi)\right]d\overline{W}_{t}^{1}\\
&~~~~~~~~~~~~~~~~+\left[\alpha_{t}\left(\left(h^{2}+\mathcal{B}^{2}\right)\phi\right)-\alpha_{t}(h^{2})\alpha_{t}(\phi)\right]d\overline{W}_{t}^{2}.
\label{alpha-section}
 \end{align}
 $\mbox{for any}~~\phi\in B(\mathbb{R}^{2})$(the space of bounded measurable functions $\mathbb{R}^{2}\to\mathbb{R}$). The operators $\mathcal{B}^{1}$ and $\mathcal{B}^{2}$ are given in (\ref{operator-B}).
\noindent Moreover the dynamics of $(\overline{\mu}_{t},\overline{\beta}_{t})$ satisfy the following stochastic differential equations: 
 \begin{align*}
&d\overline{\mu}_{t}=\alpha_{t}(a)dt + [\alpha_{t}\left(h^{1}\phi_{1}+b_{1}\right)-\alpha_{t}(h^{1})\alpha_{t}(\phi_{1})] d\overline{W}_{t}^{1}+[\alpha_{t}\left(h^{2}\phi_{1}+b_{2}\right)-\alpha_{t}(h^{2})\alpha_{t}(\phi_{1})]d\overline{W}_{t}^{2},\\
&d\overline{\beta}_{t}=\alpha_{t}(\overline{a}) dt+ [\alpha_{t}\left(h^{1}\phi_{2}+\overline{b}_{1}\right)-\alpha_{t}(h^{1})\alpha_{t}(\phi_{2})]d\overline{W}_{t}^{1} + [\alpha_{t}\left(h^{2}\phi_{2}+\overline{b}_{2}\right)-\alpha_{t}(h^{2})\alpha_{t}(\phi_{2})] d\overline{W}_{t}^{2}.
\end{align*}
 
\end{proposition}

\begin{proof}
From the definition of $\tilde{\mu}_{t}$ and $\tilde{\beta}_{t}$ and depending on the models of $\mu_{t}$ and $\beta_{t}$, we have from It\^{o}'s formula that $V_{t}$ still appear in the dynamics of $\tilde{\mu}_{t}$ and $\tilde{\beta}_{t}$. As $V_{t}=\Upsilon(\tilde{\mu}_{t},\tilde{\beta}_{t})$, then we can describe the dynamics of the signal process $X_{t}=(\tilde{\mu},\tilde{\beta}_{t})$ as in (\ref{signal-process}). On the other hand, from the definition of the observation process given by (\ref{observation-process-section}), we have that the sensor function $h=(h_{1},h_{2})$ has a linear growth condition. Thus, as assumptions $i),ii)$ and $iii)$ are verified, then we can deduce from lemma \ref{mart-proof}, that the conditions (\ref{assumption-L}) and (\ref{cond-zakai-1}) are proved. Therefore the dynamics of $\alpha_{t}$ given in (\ref{alpha-section}) is deduced from proposition \ref{general-Kushner-equation}. 

\noindent It remains to deduce the dynamics of $(\overline{\mu}_{t},\overline{\beta}_{t})$.  

\noindent Let us consider the functions $\phi_{1}$ and $\phi_{2}$ as follows: 
\[
\mbox{for}~~x=(m,b),~~~~\phi_{1}(x)=m~~\mbox{and}~~\phi_{2}(x)=b.
\]
Then the filters $\overline{\mu}_{t}$ (resp.$\overline{\beta}_{t}$) can be deduce from (\ref{alpha-section}) by replacing $\phi$ by $\phi_{1}$ (resp.$\phi_{2}$). The problem here is that the Kushner-Stratonovich equation (\ref{alpha-section}) holds for any bounded Borel measurable $\phi$. But as $\phi_{1}$ (resp.$\phi_{2}$) not bounded, we proceed by truncating of $\phi_{1}$ (resp.$\phi_{2}$) at a fixed level which we let tend to infinity. For this, let us introduce the functions $(\psi^{k})_{k>0}$ defined as
\begin{align*}
\psi^{k}(x)=\psi(x/k),~~~~~x~in~\mathbb{R}^{2},
\end{align*}  
where 
$$
\psi(x) =
\left\{
\begin{array}{lr}
1& ~~~\mbox{if}~~|x|\le 1\\
\exp(\frac{|x|^{2}-1}{|x|^{2}-4})&~~~~~~~~~~~~~~~\mbox{if}~~1<|x|<2\\
2&~~~~~~~~\mbox{if}~~|x|\ge 2.
\end{array}
\right.
$$

Then by using the following relations given in:
\begin{align*}
&\displaystyle\lim_{k\to\infty}\phi_{1}\psi^{k}(x)=\phi_{1}(x),~~~~~~~~|\phi_{1}(x)\psi^{k}(x)|\le|\phi_{1}(x)|,\\
&\displaystyle\lim_{k\to\infty}A_{s}(\phi_{1}\psi^{k})(x)=A_{s}\phi_{1}(x).
\end{align*}
Then by replacing in equation (\ref{alpha-section}) $\phi$ by $\phi_{1}\psi^{k}$ and from dominated convergence theorem, we may pass to the limit as $k\to \infty$ and then we deduce that $\overline{\mu}_{t}:=\alpha_{t}(\phi_{1})$ (resp.$\overline{\beta}_{t}:=\alpha_{t}(\phi_{2})$ ) satisfy the dynamics given above.  
\end{proof}

\subsubsection{Existence and uniqueness of the solution to equation (\ref{alpha-section})}

\smallskip\smallskip 
We now take sufficient assumption on the coefficients of the signal-observation system in order to show that equation (\ref{alpha-section}) has a unique solution, see Bain and Crisan \cite[chap.4]{Bain}. We define in the following the space within which we prove the uniqueness.

\noindent Let us define the space of measure-valued stochastic processes within which we prove uniqueness of the solution to equation (\ref{alpha-section}). This space has to be chosen so that it contains only measures with respect to which the integral of any function with linear growth is finite. The reason of this choice is that we want to allow to the coefficients of the signal and observation processes to be unbounded. 

\smallskip\smallskip
Let $\psi:\mathbb{R}^{2}\to \mathbb{R}$ be the function $\psi(x)=1+||x||$, for any $x \in \mathbb{R}^{2}$ and define $C^{l}(\mathbb{R}^{2})$ to be the space of continuous functions $\phi$ such that $\phi/\psi\in C_{b}(\mathbb{R}^{2})$(the space of bounded continuous functions).

\noindent Let us denote by $\mathbb{M}^{l}(\mathbb{R}^{2})$ the space of finite measure $\mathcal{M}$ such that $\mathcal{M}(\psi)<\infty$. In particular, this implies that $\mathcal{\mu}(\phi)<\infty$ for all $\phi \in C^{l}(\mathbb{R}^{2})$. Moreover, we endow $\mathbb{M}^{l}(\mathbb{R}^{2})$ wit the corresponding weak topology: A sequence $(\mathcal{\mu}_{n})$ of measures in $\mathbb{M}^{l}(\mathbb{R}^{2})$ converges to $\mathcal{\mu} \in \mathbb{M}^{l}(\mathbb{R}^{2})$   if and only if $\displaystyle\lim_{n\to \infty}\mathcal{\mu}_{n}(\phi)=\mathcal{\mu}(\phi)$, for all $\phi \in C^{l}(\mathbb{R}^{2})$.

\begin{definition}
\begin{itemize}
\item The Class $\mathbb{U}$ is the space of all $Y_{t}$-adapted $\mathbb{M}^{l}(\mathbb{R}^{2})$-valued stochastic process $(\mathcal{\mu})_{t\geqslant 0}$ with c\`{a}dl\`{a}g paths such that, for all $t \geqslant 0$, we have
\[
\tilde{\mathbb{E}}\left[\displaystyle\int_{0}^{t}(\mathcal{\mu}_{s}(\psi))^{2}ds\right]<\infty.
\]
\item The Class $\mathbb{\tilde{U}}$ is the space of all $Y_{t}$-adapted $\mathbb{M}^{l}(\mathbb{R}^{2})$-valued stochastic process $(\mathcal{\mu})_{t\geqslant 0}$ with c\`{a}dl\`{a}g paths such that the process $m^{\mathcal{\mu}}\mathcal{\mu}$ belongs to the class $\mathbb{U}$, where the process $m^{\mathcal{\mu}}$ is defined as:
\[
m_{t}^{\mathcal{\mu}}=\exp(\displaystyle\int_{0}^{t}\mathcal{\mu}_{s}(h^{T})dY_{s}-\dfrac{1}{2}\displaystyle\int_{0}^{t}\mathcal{\mu}_{s}(h^{T})\mathcal{\mu}_{s}(h)ds)
\]
\end{itemize}
\end{definition}

Now we state the uniqueness result of the solution to equation (\ref{alpha-section}), see theorem $4.19$ in Bain and Crisan \cite[chap.4]{Bain} 
\begin{proposition}
Assuming that the functions $A$, $K$ and $h$ defined in (\ref{notation}) have twice continuously differentiable components and all their derivatives of first and second order are bounded.  Then equation (\ref{alpha-section}) has a unique solution in the class $\mathbb{\tilde{U}}$.
\end{proposition} 

\smallskip\smallskip
\begin{remark}
The equations satisfied by the filters are infinite-dimensional and cannot be solved explicitly. These filters have to be solved numerically, but in concrete application, the filter could thus never be implemented exactly, so in order to avoid this difficulty, some approximation schemes have been proposed. For example, the extended Kalman filter, which is based upon linearization of the state equation around the current estimate, see e.g Pardoux \cite{pardoux}. This method is not mathematically justified, but it is widely used in practice. The partial differential equations method which based on the fact that the density of the unnormalised conditional distribution of the signal is the solution of a partial differential equation, see e.g Bensoussan\cite{Bensoussan} and Pardoux \cite{pardoux}. Also, we can use the approximation scheme used by Gobet el al \cite{Gobet} which consist in discretizing the Zakai equation, which is linear, and then deduce the approximation of the conditional distribution $\alpha_{t}$ from Kllianpur-Striebel formula (\ref{Kallianpur}).
\end{remark}

 \subsubsection{Application}
 In this section, we will present two types of models: a models for which we cannot apply our result in proposition \ref{dependence-drift} in order to deduce the filters estimate and a models where proposition \ref{dependence-drift} can be applied.
 
 \smallskip
\noindent Let us consider the following 
 \begin{align}
\label{model-S-app}
&\frac{dS_{t}}{S_{t}}= \mu_{t}dt + e^{V_{t}} dW^{1}_{t}\\
\label{model-V-app}
&  dV_{t}=\lambda_{V}\left(\theta- V_{t}\right) dt + \sigma_{V}\rho dW^{1}_{t}+\sigma_{V}\sqrt{1-\rho^{2}} dW^{2}_{t} \\
& d\mu_{t}=\lambda_{\mu}\left(\theta_{\mu}-\mu_{t}\right) dt + \sigma_{\mu} dW^{3}_{t},~~~~~~~~\mu_{0}\rightsquigarrow \mathcal{N}(m_{0},\sigma_{0}),
\label{model-mu-app}
\end{align} 

\noindent Here the risks of the models are given by: 
\[\tilde{\mu}_{t}=\dfrac{\mu_{t}}{e^{V_{t}}}~~~~~~~~~\tilde{\beta}_{t}=\dfrac{\lambda_{V}(\theta-V_{t})}{\sigma_{V}\sqrt{1-\rho^{2}}}-\dfrac{\rho}{\sqrt{1-\rho^{2}}}\tilde{\mu}_{t}\]
\noindent Applying It\^{o}'s formula on $\tilde{\mu}_{t}$ and $\beta_{t}$, we have the following dynamics:

\begin{align*}
\nonumber
\tilde{\mu}_{t}&=\tilde{\mu}_{0}+\displaystyle\int_{0}^{t}\lambda_{\mu}\theta_{\mu} e^{-V_{s}}ds+\tilde{\mu}_{s}\left(\sigma_{V}^{2}-\lambda_{\mu}-\left[\lambda_{V}\left(\theta-V_{s}\right)+\frac{1}{2}\sigma_{V}^{2}\right]\right)ds+\displaystyle\int_{0}^{t}\sigma_{V}e^{-V_{s}}dW^{3}_{s}\\
&~~~~~~~-\displaystyle\int_{0}^{t}\rho\sigma_{V}\tilde{\mu}_{s}dW_{s}^{1}-\displaystyle\int_{0}^{t}
\sqrt{1-\rho^{2}}\sigma_{V}\tilde{\mu}_{s}dW_{s}^{2}.\\
&\tilde{\beta}_{t}=-\displaystyle\int_{0}^{t}\frac{\lambda_{V}^{2}(\theta-V_{s})}{\sigma_{V}\sqrt{1-\rho^{2}}}ds-\displaystyle\int_{0}^{t}\frac{\lambda_{V}\rho}{\sqrt{1-\rho^{2}}}dW^{1}_{s}-\displaystyle\int_{0}^{t}\lambda_{V} dW_{s}^{2}
-\frac{\rho}{\sqrt{1-\rho^{2}}}d\tilde{\mu}_{s}.
\end{align*}

\noindent On the other hand, from the definition of $\tilde{\beta}_{t}$, we can express $V_{t}$ in terms of $\tilde{\mu}_{t}$ and $\tilde{\beta}_{t}$ as follows:
\begin{equation}\label{description}
V_{t}=-\frac{\sigma_{V}\sqrt{1-\rho^{2}}}{\lambda_{V}}\tilde{\beta}_{t}-\frac{\sigma_{V}\rho}{\lambda_{V}}\tilde{\mu}_{t}+\theta.
\end{equation}
If we replace $V_{t}$ in the above dynamics, we can deduce that $(\tilde{\mu}_{t},\tilde{\beta}_{t})$ can be described as in (\ref{signal-process}), where:

\begin{align*}
&a(m,b)=\lambda_{\mu} \theta_{\mu}\exp\left(\frac{\sigma_{V}\sqrt{1-\rho^{2}}}{\lambda_{V}}b + \frac{\sigma_{V}\rho}{\lambda_{V}}m-\theta\right)+ \left(\frac{1}{2}\sigma_{V}^{2}-\lambda_{\mu}-\sigma_{V}\sqrt{1-\rho^{2}}b-\sigma_{V}\rho m\right)m;\\
&b_{1}(m,b)=-\rho\sigma_{V}m;~~b_{2}(m,b)=-\sigma_{V}m \sqrt{1-\rho^{2}};~~g_{1}(m,b)=\sigma_{V}\exp\left(\frac{\sigma_{V}\sqrt{1-\rho^{2}}}{\lambda_{V}}b + \frac{\sigma_{V}\rho}{\lambda_{V}}m-\theta\right).
\end{align*}
and 
\begin{align*}
&\overline{a}(m,b)=-\lambda_{V} b -\frac{\lambda_{V} \rho}{\sqrt{1-\rho^{2}}}m-\frac{\rho}{\sqrt{1-\rho^{2}}}a(m,b);~~\overline{b}_{1}(m,b)=-\frac{\lambda_{V} \rho}{\sqrt{1-\rho^{2}}}-\frac{\rho^{2}\sigma_{V}}{\sqrt{1-\rho^{2}}}m\\
&\overline{b}_{2}(m,b)=-\lambda_{V}+\rho\sigma_{V}m;~~\overline{g}_{1}(m,b)= \frac{-\rho}{\sqrt{1-\rho^{2}}}\exp\left(\frac{\sigma_{V}\sqrt{1-\rho^{2}}}{\lambda_{V}}b + \frac{\sigma_{V}\rho}{\lambda_{V}}m-\theta\right); g_{2}=\overline{g}_{2}=0.
\end{align*} 

With (\ref{description}), the dynamics of $(\tilde{\mu}_{t},\tilde{\beta}_{t})$ is described as in (\ref{signal-process}) but assumption $i)$ about the globally Lipschitz conditions is not satisfied, then proposition \ref{dependence-drift} can't be applied. 

\begin{remark}
Notice that here $\beta_{t}$ is a constant function. Also we can choose for example $\beta_{t}=\mu_{t}$ which in this case we can still describe $V_{t}$ only in terms of $\tilde{\mu}_{t}$ and $\tilde{\beta}_{t}$.  But if we take $\beta_{t}$ is another process, in this case it is not clear that $V_{t}$ can be described only in terms of  $\tilde{\mu}_{t}$ and $\tilde{\beta}_{t}$.
\end{remark}

\noindent Let us now consider another example: Heston model 
  \begin{align*}
&\frac{dS_{t}}{S_{t}}= \mu_{t}dt + \sqrt{V_{t}}dW^{1}_{t},\\
& dV_{t}=\lambda_{V}\left(\theta- V_{t}\right) dt + \sigma_{V}\sqrt{V_{t}}\left(\rho dW^{1}_{t}+\sqrt{1-\rho^{2}} dW^{2}_{t}\right), \\
&d\mu_{t}=\lambda_{\mu}\left(\theta_{\mu}-\mu_{t}\right) dt + \sigma_{\mu} dW^{3}_{t},~~~~~~\mu_{0}\rightsquigarrow \mathcal{N}(m_{0},\sigma_{0}),\\
\end{align*}
 Here the risks are given by $\tilde{\mu}_{t}=\dfrac{\mu_{t}}{\sqrt{V_{t}}}$ and $\tilde{\beta}_{t}=\dfrac{\lambda_{V}(\theta-V_{t})}{\sigma_{V}\sqrt{V_{t}}\sqrt{1-\rho^{2}}}-\dfrac{\rho}{\sqrt{1-\rho^{2}}}\tilde{\mu}_{t}.$
 Also here we are in the above situation that is we can describe the dynamics of $\tilde{\mu}_{t}$ and $\tilde{\beta}_{t}$ as in (\ref{dependence-drift}), but assumption $i)$ is not satisfied.

 \vspace{3mm}
Now we give some examples with which proposition (\ref{dependence-drift}) can be applied and therefore we can deduce the filters estimate. 
we will be interested by the stochastic factor Garch model and the stochastic factor Log Ornstein-Uhlenbeck model.

\vspace{3mm}
\noindent \underline{Stochastic factor Garch model:}

\noindent Let us consider the following Garch-model: 
\begin{align*}
&\frac{dS_{t}}{S_{t}}= \sqrt{V_{t}}\left(\mu_{t}dt + dW^{1}_{t}\right),\\
& dV_{t}=\beta_{t}\left(\theta- V_{t}\right) dt + \sigma_{V}V_{t}\left(\rho dW^{1}_{t}+\sqrt{1-\rho^{2}} dW^{2}_{t}\right), \\
&d\mu_{t}=\lambda_{\mu}\left(\theta_{\mu}-\mu_{t}\right) dt + \sigma_{\mu} dW^{3}_{t},~~~~~~\mu_{0}\rightsquigarrow \mathcal{N}(m_{0},\sigma_{0}),\\
& d\beta_{t}= \lambda_{\beta}\beta_{t} dt + \sigma_{\beta} dW^{4}_{t}~~~~~~\beta_{0}\rightsquigarrow \mathcal{N}(m_{1},\sigma_{1}).
\end{align*}
where $W^{3}$ and $W^{4}$ are independent and independent from $W^{1}$ and $W^{2}$ where $\mu_{0}$ and $\beta_{0}$ follow respectively a normal distribution of mean $m_{0}$ (resp.$m_{1}$) and variance $\sigma_{0}$ (resp.$\sigma_{1}$).

Here the risk of the model are given by:
\[
\tilde{\mu}_{t}=\mu_{t}~~~~~~~~\mbox{and}~~~~\tilde{\beta}_{t}=\dfrac{\beta_{t}(\theta-V_{t})}{\sqrt{1-\rho^{2}}V_{t}}-\dfrac{\rho}{\sqrt{1-\rho^{2}}}\tilde{\mu}_{t}.
\]

In order to compute the filters estimate in this case of models, we will be interested by using proposition \ref{dependence-drift}. For that, we need to take $\theta=0$. Because, if we apply It\^{o}'s formula on $\tilde{\mu}_{t}$ and $\tilde{\beta}_{t}$ in the case where $\theta\neq 0$, we obtain a dynamics with coefficients are not Lipschitz, that is, assumption $i)$ is not verify and therefore proposition \ref{dependence-drift} can't be applied. For that we will take $\theta=0$. 
Let $\theta=0$, then from It\^{o}'s formula, we have:
\begin{align*}
d\left(
\begin{array}{c}
\tilde{\mu}_{t} \\
\tilde{\beta}_{t}
\end{array} \right)= A\left(
\begin{array}{c}
\tilde{\mu}_{t} \\
\tilde{\beta}_{t}
\end{array}\right)dt+ G\left(
\begin{array}{c}
\tilde{\mu}_{t} \\
\tilde{\beta}_{t}
\end{array}\right)dM_{t}.
\end{align*}
where the functions $A,G$ and $B$ are given as follows:
\[ A\left(
\begin{array}{c}
m\\
b
\end{array}\right)=\left(
\begin{array}{c}
\lambda_{\mu}(\theta_{\mu}-m)\\
\lambda_{\beta} b +\dfrac{\rho(\lambda_{\beta}+\lambda_{\mu})}{\overline{\rho}} m -\dfrac{\rho\lambda_{\mu}\theta_{\mu}}{\overline{\rho}}
\end{array} \right),
G\left(
\begin{array}{c}
m\\
b
\end{array}\right)=\left(
\begin{array}{cc}
-\dfrac{\rho \sigma_{\mu}}{\overline{\rho}} & 0\\
0 & \sigma_{\beta}(b+\dfrac{\rho}{\overline{\rho}}m)
\end{array} \right).
\]
where $\overline{\rho}=\sqrt{1-\rho^{2}}$ and the function $B$ is null, so we are in the case where the signal process $X_{t}:=(\tilde{\mu}_{t},\tilde{\beta_{t}})$ and the observation processes $Y_{t}:=(\tilde{W}^{1}_{t},\tilde{W}^{2}_{t})$ are independent. This implies that the operator $\mathcal{B}^{1}$ and $\mathcal{B}^{2}$ will disappear in the Zakai and Kushner-Stratonovich equations. As for this model, the assumptions of proposition \ref{dependence-drift} are satisfied, then the conditional distribution $\alpha_{t}$ is given for any $\phi$ by:
\[
d\alpha_{t}(\phi)=\alpha_{t}(\mathcal{A}\phi)dt +\left[\alpha_{t}\left(h^{1}\phi\right)-\alpha_{t}(h^{1})\alpha_{t}(\phi)\right]d\overline{W}_{t}^{1}+\left[\alpha_{t}\left(h^{2}\phi\right)-\alpha_{t}(h^{2})\alpha_{t}(\phi)\right]d\overline{W}_{t}^{2}.
\] 
Here the operator $\mathcal{A}$ is given by (\ref{Operator-A}), where $K=\dfrac{1}{2}GG^{T}$.

\noindent Therefore, the dynamics of the filter estimate are given as follows: 
\begin{align*}
&d\overline{\mu}_{t}= \lambda_{\mu}(\theta_{\mu}-\overline{\mu}_{t})dt + \left(\alpha_{t}(h^{1}\phi_{1})-\overline{\mu}_{t}^{2}\right)d\overline{W}
^{1}_{t}+ \left(\alpha_{t}(h^{2}\phi_{1})-\overline{\beta_{t}}\overline{\mu}_{t}\right)d\overline{W}^{2}_{t},\\
& d\overline{\beta}_{t}= \left(\lambda_{\beta} \overline{\beta}_{t}+\dfrac{\rho(\lambda_{\beta}+ \lambda_{\mu})}{\overline{\rho}}\overline{\mu}_{t}-\dfrac{\rho \lambda_{\mu}\theta_{\mu}}{\overline{\rho}}\right) dt + \left(\alpha_{t}(h^{1}\phi_{2})-\overline{\mu}_{t}\overline{\beta}_{t}\right)d\overline{W}
^{1}_{t}+ \left(\alpha_{t}(h^{2}\phi_{2})-\overline{\beta_{t}}^{2}\right)d\overline{W}^{2}_{t}.
\end{align*}
Numerically, in order to simulate $\alpha_{t}$, we can use the approximation scheme developed by Gobet et al \cite{Gobet} or the extended Kalman filter studied by Pardoux \cite[Chap.6]{pardoux}. 

\vspace{4mm}

Also we consider another example for which we can apply proposition (\ref{dependence-drift}): the stochastic factor Log Ornstein-Uhlenbeck model. 
the special features of this model is not only we can apply proposition (\ref{dependence-drift}), but also  we are in a particular case of the signal-observation system (\ref{signal-process}) where $A,B$ and $G$ are deterministic. So we are in the framework of the classical Kalman-Bucy filter with correlation between the signal and the observation processes, see Pardoux \cite[Chap.6]{Pardoux-stoch} and Kallianpur\cite[Theo 10.5.1]{Kallianpur}. This filter is deduced from the general Kushner-Stratonovich equation (\ref{general-Kushner-equation}), but the advantage of this filter is that it is a finite dimensional filter.  

\vspace{2mm}
\noindent \underline{\textbf{Finite dimensional filter:}} stochastic factor Log Ornstein-Uhlenbeck model

\noindent Let us consider the following Log Ornstein-Uhlenbeck model:
\begin{align}
\label{model-S-app}
&\frac{dS_{t}}{S_{t}}= e^{V_{t}}\left(\mu_{t}dt + dW^{1}_{t}\right)\\
\label{model-V-app}
&  dV_{t}=\lambda_{V}\left(\theta- V_{t}\right) dt + \sigma_{V}\rho dW^{1}_{t}+\sigma_{V}\sqrt{1-\rho^{2}} dW^{2}_{t} \\
& d\mu_{t}=\lambda_{\mu}\left(\theta_{\mu}-\mu_{t}\right) dt + \sigma_{\mu} dW^{3}_{t}. 
\label{model-mu-app}
\end{align}  
Then from the definition of $\tilde{\mu}_{t}$ and $\tilde{\beta}_{t}$ and It\^{o}'s formula, the risks of the system have the following dynamics:
$$d\left(
\begin{array}{c}
\tilde{\mu}_{t} \\
\tilde{\beta}_{t}
\end{array} \right)= \left(A(t)\left(
\begin{array}{c}
\tilde{\mu}_{t} \\
\tilde{\beta}_{t}
\end{array} \right)+b(t)\right)dt+ G(t)d\left(
\begin{array}{c}
\overline{W}^{3}_{t} \\
\overline{W}^{4}_{t}
\end{array}\right)+ B(t) \left(
\begin{array}{c}
\overline{W}^{1}_{t} \\
\overline{W}^{2}_{t}
\end{array} \right).$$
Here:

\[ A=\left(
\begin{array}{cc}
 -\lambda_{\mu}\ & 0\\
\dfrac{\rho[\lambda_{\mu}-\lambda_{V}]}{\overline{\rho}} & -\lambda_{V}
\end{array} \right),~b=\left(
\begin{array}{c}
\lambda_{\mu}\theta_{\mu}\\
-\dfrac{\rho}{\overline{\rho}}\lambda_{\mu}\theta_{\mu}
\end{array} \right),~G=\left(
\begin{array}{cc}
\sigma_{\mu}& 0\\
-\dfrac{\rho}{\overline{\rho}}\sigma_{\mu}& 0
\end{array} \right)~ B=\left(
\begin{array}{cc}
0 & 0 \\
-\dfrac{\rho}{\overline{\rho}}\lambda_{V}  &-\lambda_{V}
\end{array} \right).\]
where $\overline{\rho}=\sqrt{1-\rho^{2}}$.

Therefore using theorem $10.5.1$ in \cite{Kallianpur}, we can deduce the following stochastic differential equations for the filters: 

\begin{align}\label{filter-estimates-application}
d\left(
\begin{array}{c}
\overline{\mu}_{t} \\
\overline{\beta}_{t}
\end{array} \right)= \left(A(t)\left(
\begin{array}{c}
\overline{\mu}_{t} \\
\overline{\beta}_{t}
\end{array} \right)+b(t)\right)dt+\left(B(t)+\Theta_{t}\right)d\left(
\begin{array}{c}
\overline{W}^{1}_{t} \\
\overline{W}^{2}_{t}
\end{array} \right).
\end{align}
Where $\Theta_{t}$ is the conditional covariance matrix ($2\times 2$) of the signal satisfies the following deterministic matrix Ricatti equation:
\begin{equation}\label{Riccati-equation}
d\Theta_{t}=A\Theta_{t}+\Theta_{t}A^{T}+ GG^{T}-\Theta_{t}\Theta_{t}^{T}-\Theta_{t}B^{T}-B\Theta_{t}.
\end{equation}

Also we can consider the case where the mean $\theta$ of the stochastic volatility $V_{t}$ is a linear function of $\mu_{t}$. For example, assume the above dynamics of $(S_{t},V_{t},\mu_{t})$ with $\theta=\mu_{t}$. Therefore, the filters estimate $(\overline{\mu}_{t},\overline{\beta}_{t})$ verifies (\ref{filter-estimates-application}). Here $G$ and $B$ are the same matrix given above, but $A$ and $b$ are given by:
 
  \[ A=\left(
\begin{array}{cc}
 -\lambda_{\mu}\ & 0\\
\dfrac{\rho[\lambda_{\mu}-\lambda_{V}]}{\overline{\rho}}-\dfrac{\lambda_{\mu}\lambda_{V}}{\sigma_{V}\overline{\rho}}& -\lambda_{V}
\end{array} \right),~b=\left(
\begin{array}{c}
\lambda_{\mu}\theta_{\mu}\\
\dfrac{\lambda_{V}}{\sigma_{V}\overline{\rho}}-\dfrac{\rho}{\overline{\rho}}\lambda_{\mu}\theta_{\mu}
\end{array} \right).\]

\begin{remark}
Also, we have the above results about the filters estimate if we  consider the Stien-stein model, where the stock  has the dynamics: $\frac{dS_{t}}{S_{t}}= |V_{t}|\left(\mu_{t}dt + dW^{1}_{t}\right)$ and the stochastic volatility $V_{t}$ and the drift $\mu_{t}$ are given by (\ref{model-V-app}) and (\ref{model-mu-app}).
\end{remark}

\section{Application to portfolio optimization}\label{optimization-section}
Before presenting our results, let us recall that the trader's objective is to solve the following optimization problem: 
\begin{equation}\label{rappel-Optim-prob} 
J(x)=\displaystyle\sup_{\pi\in \mathcal{A}_{t}} \mathbb{E}[U(R_{T}^{\pi})]~~~~ x>0,
\end{equation}
  where  the dynamics of $R_{t}^{\pi}$ in the full information context is given by:
\[
dR_{t}^{\pi} =  R_{t}^{\pi}\pi_{t}\left(g(V_{t})~\overline{\mu}_{t}dt + g(V_{t})d \overline{W}^{1}_{t}\right).
\]
Here $\mathcal{A}_{t}$ is the set of admissible controls $\pi_{t}$ which are $\mathbb{F}^{S}$-adapted process,take their value in a compact $\mathbb{U}\subset \mathbb{R}$, and satisfies the integrability condition: 

\begin{equation}\label{integrability-cond-1}
\displaystyle\int_{t}^{T} g^{2}(V_{s})\pi_{s}^{2} <\infty~~~~~~\mathbb{P}~a.s.
\end{equation}

\vspace{2mm}
 We have showed that using the nonlinear filtering theory, the partial observation portfolio problem is transformed into a full observation one with the additional filter in the dynamic of the wealth, for which one may apply the martingale or PDE approach. 

\smallskip
\noindent Here we will interested by the martingale approach in order to resolve our optimization problem. The motivation to use the martingale approach instead of the PDE approach is that we don't need to impose any constraint on the admissible control (see remark \ref{pas-PDE-approach}).

\smallskip
\noindent As the reduced market model is not complete, due to the stochastic factor $V$, we have to solve the related dual optimization problem. For that, we complement the martingale approach by using the PDE approach in order to solve explicitly the dual problem. For the case of CARA's utility functions, show by verification result, that under some assumptions on the market coefficients, the dual value function and the dual optimizer are related to the solution of a semilinear partial differential equation.

\subsection{\textbf{Martingale approach}}
Before presenting our result concerning the solution of the dual problem, let us begin by reminding some general results about the martingale approach.

\noindent  The martingale approach in incomplete market is based on a dual formulation of the optimization problem in terms of a suitable family of $(\mathbb{P},\mathbb{G})$-local martingales. 
The important result for the dual formulation is the martingale representation 
theorem given in \cite{Pham-Quenez} for $(\mathbb{P},\mathbb{G})$-local martingales with respect to the innovation processes $\overline{W}^{1}$ and $\overline{W}^{2}$. 

 \begin{lemma}[Martingale representation theorem]
Let $A$ be any $(\mathbb{P},\mathbb{G})$-local martingale. Then, there exist a $\mathbb{G}$-adapted processes $\phi$ and $\psi$, $\mathbb{P}$ a.s. square-integrable and such that
\begin{equation}\label{representation}
A_{t}=\displaystyle\int_{0}^{t}\phi_{s}d\overline{W}^{1}_{s}+
\displaystyle\int_{0}^{t}\psi_{s}d\overline{W}^{2}_{s}.
\end{equation}  
\end{lemma}

\noindent Now, we aim to describe the dual formulation of the optimization problem. 
We now make the following assumption which will be useful in the sequel: 
\begin{equation}\label{assumption-density}
\displaystyle\int_{0}^{T}\overline{\mu}_{t}^{2}dt < \infty,~~~~~~\displaystyle\int_{0}^{T}\nu_{t}^{2}dt < \infty~~~~\mathbb{P}-a.s.
\end{equation}
\noindent For any $\mathbb{G}$-adapted process $\nu=\{\nu_{t},~0 \leq t \leq T\}$, which satisfies (\ref{assumption-density}), we introduce the $\left(\mathbb{P},\mathbb{G}\right)$-local martingale
strictly positive:
\begin{equation}\label{exp-EMM}
Z_{t}^{\nu}=\exp\left(-\displaystyle\int_{0}^{t}\overline{\mu}_{s}d\overline{W}^{1}_{s}-
\displaystyle\int_{0}^{t}\nu_{s}d\overline{W}^{2}_{s}
-\frac{1}{2}\displaystyle\int_{0}^{t}\overline{\mu}_{s}^{2} ds-\frac{1}{2}\displaystyle\int_{0}^{t}\nu_{s}^{2} ds\right)
\end{equation}

\noindent When, $\mathbb{E}\left[Z_{T}^{\nu}\right]=1$, the process $Z$ is a martingale and then there exists a probability measure $\mathbb{Q}$ equivalent to $\mathbb{P}$ with:
\begin{equation*}
\frac{d\mathbb{Q}}{d\mathbb{P}}|_{\mathcal{G}_{t}}=Z_{T}^{\nu}.
\end{equation*}

\noindent Here $\overline{\mu}$ is the risk related to the asset's Brownian motion $W^{1}$, which is chosen such that $Q$ is a equivalent martingale measure, that is, the process $Z^{\nu} R$ is a $\left(\mathbb{P},\mathbb{G}\right)$-local martingale. On the other hand, $\nu$ is the risk related to the stochastic volatility's Brownian motion and this risk will be determined as the optimal solution of the dual problem defined below. 

\noindent Consequently, from It\^{o}'s formula, the process $Z^{\nu}$ satisfies:
\begin{equation}\label{risk-process}
dZ_{t}^{\nu}=-Z_{t}^{\nu}\left(\overline{\mu}_{s}d\overline{W}_{s}^{1}+\nu_{s}d\overline{W}^{2}_{s}\right).
\end{equation}  

\vspace{5mm}
\noindent As shown by Karatzas et al \cite{Ioannis and Lehoczky}, the solution of the primal problem (\ref{rappel-Optim-prob}) relying upon solving the dual optimization problem:

\begin{equation}\label{dual-problem}
J_{dual}(z)=\displaystyle\inf_{\mathbb{Q}\in \mathcal{Q}} \mathbb{E}\left[\tilde{U}(z\frac{d\mathbb{Q}}{d\mathbb{P}})\right]:=\inf_{\nu \in \mathcal{K}}\mathbb{E}\left[\tilde{U}(zZ_{T}^{\nu})\right] ,~~~~~~~~z>0
\end{equation}

\noindent Where:

\begin{itemize}
\item $\mathcal{Q}$ is the set of equivalent martingale measures given by:
\begin{equation}\label{equivalent-measure}
\mathcal{Q}=\{\mathbb{Q}\sim \mathbb{P}|~R~\mbox{is a local}~ (\mathbb{Q},\mathbb{G})-\mbox{martingale}\}.
\end{equation}
\item   $\tilde{U}$ is the convex dual of $U$ given by:
\begin{equation}\label{convex}
\tilde{U}(y)=\displaystyle\sup_{m>0}\left[U(m)-ym\right],~~~~~~~~m>0.
\end{equation}
\item  $\mathcal{K}$ is the Hilbert space of $\mathbb{G}$-adapted process $\nu$ such that $\mathbb{E}\left[\displaystyle \int_{0}^{T}|\nu_{t}^{2}|dt \right] < \infty$.
\end{itemize}

\noindent We henceforth impose the following assumptions on the utility functions in order to guarantee that the dual problem admits a solution $\tilde{\nu} \in \mathcal{K}$:

\begin{assumption}\label{existence-dual}
\begin{itemize}

\item For some $p\in(0,1), \gamma\in(1,\infty)$, we have
\begin{equation*}
p U'(x)\geq U'(\gamma x)~~~~~~~~\forall x \in (0,\infty).
\end{equation*}
\item $x\to xU'(x)$ is nondecreasing on $(0,\infty)$.
\item For every $z \in (0,\infty)$, there exists $\nu \in \mathcal{K}$ such that $\tilde{J}(z)< \infty$.
\end{itemize}
\end{assumption}
By same arguments as in theorem $12.1$ in  Karatzas et al \cite{Ioannis and Lehoczky}, we have existence to the dual problem (\ref{dual-problem}). 
 \begin{proposition}
 Under assumption \ref{existence-dual}, for all $z>0$, the dual problem (\ref{dual-problem}) admits a solution $\tilde{\nu}(z) \in \mathcal{K}$
 \end{proposition}

\noindent In the sequel, we denote by $I:]0,\infty[\rightarrow]0,\infty[$ the inverse function of $U'$ on $]0,\infty[$. It's a decreasing function and verifies $\displaystyle\lim_{x\to 0 } I(x)=\infty $ and $\displaystyle\lim_{x\to\infty}I(x)=0$.

\noindent Now from Karatzas et al \cite{Ioannis and Lehoczky} and Owen \cite{Owen}, we have the following result about the solution of the primal utility maximization problem (\ref{Optim-prob-1}).  

\begin{theorem}\label{optimal-wel}
The optimal wealth for the utility maximization problem (\ref{Optim-prob-1}) is given by 
\[
\tilde{R}_{t}=\mathbb{E}\left[\frac{Z_{T}^{\tilde{\nu}}}{Z_{t}^{\tilde{\nu}}}I(z_{x}Z_{T}^{\tilde{\nu}})|\mathcal{G}_{t}\right]
\]
where $ \tilde{\nu} =\tilde{\nu}(z_{x})$ is the solution of the dual problem and $z_{x}$ is the Lagrange multiplier such that $\mathbb{E}\left[Z_{T}^{\tilde{\nu}}I(z_{x}Z_{T}^{\tilde{\nu}})\right]=x$. Also the optimal portfolio $\tilde{\pi}$ is implicitly determined by the equation
\begin{equation}\label{optimal-portf}
d\tilde{R}_{t}=\tilde{\pi}_{t}g(V_{t})d\tilde{W}^{1}_{t}.
\end{equation}
\end{theorem}

\begin{remark}\label{Lagrange}
The constraint  $\mathbb{E}\left[Z_{T}^{\tilde{\nu}}I(z_{x}Z_{T}^{\tilde{\nu}})\right]=x$ to choose $z_{x}$ is satisfied if
\begin{equation}\label{condition-zx}
z_{x}\in argmin_{z>0}\{J_{dual}(z)+xz\}.
\end{equation} 
\end{remark}

\vspace{3mm}
Now we begin by presenting our results about the solution of the dual problem. 
\subsubsection{Solution of the dual problem (\ref{dual-problem})}

\vspace{3mm}
\noindent We remark from theorem \ref{optimal-wel} that optimal wealth depends on the optimal choice of $\nu$. So we are interested in the following by finding the optimal risk $\nu$ which is solution of (\ref{dual-problem}). 

\vspace{2mm}
\noindent Here we present two cases. Firstly, we show that in the case when the filter estimate of the price risk $\overline{\mu}_{t}\in \mathcal{F}_{t}^{\tilde{W}^{1}}$, the infimum  of the dual problem is reached for $\tilde{\nu}=0$. Secondly, for the general case, the idea is to derive a Hamilton-Jacobi-Bellman equation for dual problem, which involves the volatility risk $\nu$ as control process. 

\smallskip\smallskip

\begin{lemma}\label{EMM-opt}
Assume that $\overline{\mu}_{t}\in \mathcal{F}_{t}^{\tilde{W}^{1}}$, then the infimum of the dual problem is reached for $\tilde{\nu}=0$, that is: 
\begin{equation}
J_{dual}(z)= \displaystyle\inf_{\mathbb{Q}\in \mathcal{Q}} \mathbb{E}\left[\tilde{U}(z\frac{d\mathbb{Q}}{d\mathbb{P}})\right] =\mathbb{E}\left[\tilde{U}\left(z Z_{T}^{0}\right)\right].
\end{equation} 
\end{lemma}

\begin{proof}
See Appendix A.
\end{proof}

\vspace{4mm}
\noindent Generally, the filter estimate of the price risk doesn't satisfy lemma \ref{EMM-opt} and therefore it's a difficult problem to derive an explicit characterization for the solution of the dual problem and therefore for the optimal wealth and portfolio. For that, we need to present the dual problem as a stochastic control problem with controlled process $Z_{t}^{\nu}$ and control process $\nu$.

\vspace{4mm}
\noindent  Firstly, from the underlying dynamics of $Z_{t}^{\nu}$, we notice that our optimization problem $\ref{dual-problem}$ has three state variables which will be take in account to describe the associated Hamilton-Jaccobi-Belleman equation: the dynamic (\ref{risk-process}) of $Z_{t}^{\nu}$, the dynamic of the stochastic volatility $(V_{t})$ which is given in system $(Q)$ and the dynamic of the filter estimate of the price risk $\overline{\mu}_{t}$.

\begin{remark}\label{choice of beta}
We have showed in filtering section, that the filter estimate $\overline{\mu}_{t}$ satisfies a stochastic differential equation which in general is infinite dimensional and is not a Markov process. Therefore, we can't use it to describe the HJB.
On the other hand, we have also showed that for some models of stochastic volatility models, we can obtain a finite dimensional stochastic differential equation for $\overline{\mu}_{t}$ which is also a Markov process. So in the sequel, we will assume that the filter $\mu_{t}$ is Markov.

On the other hand, we need in general to take in account the dynamics of $\overline{\mu}_{t}$ and $\overline{\beta_{t}}$.  But for simplification, we will consider $\beta_{t}$ as a linear function of $\mu_{t}$ or a constant. Also for this choice of $\beta_{t}$, we can obtain, due to the separation technique used in proposition (\ref{explicit-form-exmaple}), a closed form for the value function and the optimal portfolio.
\end{remark}

In the following, we assume that $\overline{\mu}_{t}$ is Markov. So for initial time $t\in[0,T]$ and for fixed $z$, the dual value function is defined by the following stochastic control problem:  
\begin{equation}\label{dual-control-problem}
\tilde{J}(z,t,\overline{z},v,m):=\inf_{\nu\in \mathcal{K}}\mathbb{E}\left[\tilde{U}( z Z_{T}^{\nu})|Z_{t}^{\nu}=\overline{z},V_{t}=v,\overline{\mu}_{t}=m\right].
\end{equation}

\noindent Where the dynamics of $(Z_{t}^{\nu},V_{t},\overline{\mu}_{t})$ are given as follows:
\begin{align*}
&dZ_{t}^{\nu}=-Z_{t}^{\nu} \overline{\mu}_{s}d\overline{W}_{s}^{1} -Z_{t}^{\nu}\nu_{s}d\overline{W}^{2}_{s}\\
&dV_{t}=f(\overline{\mu}_{t},V_{t}) dt +\rho k(V_{t})d\overline{W}_{t}^{1}+\sqrt{1-\rho^{2}}k(V_{t})d\overline{W}^{2}_{t}\\ 
&d\overline{\mu}_{t}=\tau(\overline{\mu}_{t})dt + \vartheta(\overline{\mu}_{t})d\overline{W}_{t}^{1}+\Upsilon(\overline{\mu}_{t})
 d\overline{W}_{t}^{2}.
\end{align*}
where $f$ is a linear function. 

\smallskip
\noindent Remark that the dual value function in (\ref{dual-problem}) is simply deduced from $J_{dual}(z)=J_{dual}(z,0,\overline{z},v,m)$.

\vspace{3mm}

If we assume that $Y_{t}=(V_{t},\overline{\mu}_{t})$ be a bi-dimensional process, then the controlled process $(Z_{t}^{\nu},Y_{t})$ satisfies the following dynamics:
\begin{align}\label{general-dynamic-R-dual}
&dZ_{t}^{\nu}=-Z_{t}^{\nu} \psi(Y_{s})d\overline{W}_{s}^{1} -Z_{t}^{\nu}\nu_{s}d\overline{W}^{2}_{s}\\
&dY_{t}=\Gamma(Y_{t})dt +\Sigma(Y_{t})dW_{t}
\label{general-dynamic-Y-dual}
\end{align} 
where $W_{t}=(\overline{W}^{1}_{t},\overline{W}_{t}^{2})$ is a bi-dimensional Brownian motion, and for $y=(v,m)$, we have:
\begin{align*}
\psi(y)=m,~~\Gamma(y)= \left(
\begin{array}{c}
f(m,v) \\
\tau(m)
\end{array} \right)
\end{align*}and 

\[\Sigma(y)= \left(
\begin{array}{cc}
\rho k(v) & \sqrt{1-\rho^{2}}~k(v)\\
\vartheta(m) & \Upsilon(m)
\end{array} \right).\]

Then we have the new reformulation of the above stochastic problem (\ref{general-optimization}) and its HJB equation as follows: 
\begin{equation}\label{-refo-dual-control-problem}
J_{dual}(z,t,\overline{z},y):=\inf_{\nu\in \mathcal{K}}\mathbb{E}\left[\tilde{U}( z Z_{T}^{\nu})|Z_{t}^{\nu}=\overline{z},Y_{t}=y\right].
\end{equation}

\vspace{3mm}
Now assuming that $\tilde{U}$ satisfies the following property: 
\begin{equation}\label{property}
\tilde{U}(\lambda x)= g_{1}(\lambda)\tilde{U}(x)+g_{2}(\lambda), 
\end{equation}
for $\lambda>0$~and for any functions $g_{1}$ and $g_{2}$.  

\smallskip\smallskip
\noindent The special advantage of this assumption is: we can solve the dual problem (\ref{dual-problem}) independently of $z$. In general, a solution to the dual problem (\ref{dual-problem}) depends on $z$, but for this type of $\tilde{U}$ this dependence vanishes. Then (\ref{dual-problem}) reads as follows:  
\[J_{dual}(z)= 
g_{1}(z)\displaystyle\inf_{\nu\in \mathcal{K}}\mathbb{E}\left[\tilde{U}(Z_{T}^{\nu})|Z_{t}^{\nu}=\overline{z},Y_{t}=y\right]+g_{2}(z).
\] 
Let us now denote 
\begin{equation}\label{new-dual}
\tilde{J}(t,\overline{z},y)=\displaystyle\inf_{\nu\in \mathcal{K}}\mathbb{E}\left[\tilde{U}(Z_{T}^{\nu})|Z_{t}^{\nu}=\overline{z},Y_{t}=y\right].
\end{equation}
Remark that the solution of the dual problem (\ref{dual-problem}) is given by:
\begin{equation}
J_{dual}(z):=J_{dual}(z,0,\overline{z},v,m)=g_{1}(z)\tilde{J}(0,\overline{z},y)+g_{2}(z).
\end{equation}

\noindent Formally, the Hamilton-Jacobi-Bellman equation associated to the above stochastic control problem (\ref{new-dual}) is the following nonlinear partial differential equation: 
\begin{align}
\nonumber
\frac{\partial \tilde{J}}{\partial_{t}}&+\frac{1}{2}Tr\left(\Sigma(y) \Sigma ^{T}(y) D^{2}_{y}\tilde{J}\right)+\Gamma^{T}(y)D_{y}\tilde{J}\\
&+\displaystyle\inf_{\nu \in \mathcal{K}}\left[\frac{1}{2}(\psi(y)^{2}+\nu^{2})\overline{z}^{2}D^{2}_{\overline{z}}\tilde{J}-\overline{z}[\psi(y)K_{1}^{T}(y)+\nu K_{2}^{T}] D^{2}_{\overline{z},y}\tilde{J}\right]=0.
\label{new-HJB-dual}
\end{align}
with the boundary condition
\begin{equation}\label{boundary-condition-dual}
\tilde{J}(T,x,y)=U(x).
\end{equation}

And the associated optimal dual optimizer $\tilde{\nu}$ is given by:
\[
\tilde{\nu}_{t}=\dfrac{K_{2}^{T}~D^{2}_{\overline{z},y}\tilde{J}}{\overline{z}~D^{2}_{\overline{z}}\tilde{J}}.
\]
Here $D_{y}$ and $D^{2}_{y}$ denote the gradient and the Hessian operators with respect to the variable $y$. $D^{2}_{\overline{z},y}$ is the second derivative vector with respect to the variables $\overline{z}$ and $y$ and for $y=(v,m)$, $K_{1}(y)=\left(
\begin{array}{c}
\rho k(v) \\
\vartheta(m) 
\end{array} \right)$ ~~and~~$K_{2}(y)=\left(
\begin{array}{c}
\sqrt{1-\rho^{2}} k(v) \\
\Upsilon(m) 
\end{array} \right).$ 

\smallskip
\noindent The above HJB is nonlinear, but if we consider the case of CARA's utility functions and via a suitable transformation, we can make this equation semilinear and then characterize the dual value function $\tilde{J}$ through the classical solution of this semilinear equation which is more simpler than the usual fully nonlinear HJB equation.

\subsection{Special cases for utility function}
Let us consider the two more standard utility functions: logarithmic and power, defined by:

\[
U(x)= \left\{
\begin{array}{cc}
\ln(x) & x \in \mathbb{R}^{+}\\\\
\dfrac{x^{p}}{p} & ~~~~~~~~x\in \mathbb{R}^{+}, p\in(0,1) 
\end{array} \right.
\]
For these functions, the convex dual functions associated are given by: \[
\tilde{U}(z)= \left\{
\begin{array}{cc}
-(1+\ln(z)) & z \in \mathbb{R}^{+}\\\\
-\dfrac{z^{q}}{q} & ~~~~~~~~z\in \mathbb{R}, q=\dfrac{p}{p-1} 
\end{array} \right.
\]
These utility functions are of particular interests: firstly, they satisfy property (\ref{property})and secondly, due to the homogeneity of the convex dual functions together with the fact that the process  $Z_{t}^{\nu}$ and the control $\nu$ appear linearly, we can suggest a suitable transformation, for which we can characterize the dual value functions $\tilde{J}$ through a classical solution of a semilinear semilinear partial differential equations which will be described below. 

\smallskip
Let us now make some assumptions which will be useful for proving our verification results. 

\underline{\textbf{Assumption (H)}}

$i)$ $\Gamma$ and $\Sigma$ are Lipscitz and $C^{1}$ with bounded derivatives. 

$ii)$ $\Sigma\Sigma^{T}$ is uniformly elliptic, that is, there exists $c>0$ such that for $y,\xi \in \mathbb{R}^{2}$:
\begin{equation*}
\sum_{i,j=1}^{2}(\Sigma\Sigma^{T}(y))_{ij}\xi_{i}\xi_{j}\geq c |\xi|^{2}.
\end{equation*}

$iii)$ $\Sigma$ is bounded or is a deterministic matrix.  
 
$iv)$ There exists a positive constant $\epsilon$ such that
\begin{equation*}
\exp\left(\epsilon\displaystyle\int_{0}^{T}(\psi^{2}(Y_{t})+\nu_{t}^{2})dt\right) \in L^{1}(\mathbb{P}).
\end{equation*} 
 \noindent Notice that the Lipschitz assumption on $\Gamma$ and $\Sigma$ ensure the existence and uniqueness of the solution of (\ref{general-dynamic-Y-dual}). Moreover, we have:
\begin{equation}\label{integrability-condition-dual}
\mathbb{E}[\sup_{0\leq s\leq t}|Y_{s}|^{2}] < \infty.
\end{equation}

\vspace{3mm}
\noindent \subsubsection{\underline{Logarithmic utility:}}
For the logarithmic utility case, we can look for a candidate solution of (\ref{new-HJB}) and (\ref{boundary-condition}) in the form :

\begin{equation}\label{log-transf}
\tilde{J}(t,\overline{z},y)=-(1+ln(\overline{z}))-\Phi(t,y)
\end{equation}

\noindent Then direct substitution of (\ref{log-transf}) in (\ref{new-HJB-dual}) and (\ref{boundary-condition-dual}) gives us the following semilinear partial differential equation for $\Phi$:
\begin{align}\label{loga-PDE-dual}
-\frac{\partial \Phi}{\partial_{t}}-\dfrac{1}{2}Tr\left(\Sigma(y)\Sigma^{T}(y) D^{2}_{y}\Phi\right)+H(y,D_{y}\Phi)=0,
\end{align}
\noindent 
with the boundary condition:
\begin{equation}\label{loga-boundary-dual}
\Phi(T,y)=0.
\end{equation}

Where the Hamiltonian $H$ is defined by:
\begin{equation*}
 H(y,Q)=-\Gamma^{T}(y)Q+\displaystyle\inf_{\nu}\left(\dfrac{1}{2}(\psi^{2}(y)+\nu^{2})\right).
\end{equation*}
We now state a verification result for the logarithmic case, which relates the solution of the above semilinear (\ref{loga-PDE-dual}) and (\ref{loga-boundary-dual}) to the stochastic control problem (\ref{new-dual}).

\begin{theorem}[verification theorem]\label{logarithmic-verification-dual}
Let assumption \textbf{H} $i)$ holds. Suppose that there exists a solution $\Phi \in C^{1,2}([0,T)\times \mathbb{R}^{2})\cap C^{0}([0,T]\times \mathbb{R}^{2})$ to the  semilinear (\ref{loga-PDE-dual}) with the terminal condition (\ref{loga-boundary-dual}). Also we assume that $\Phi$ satisfies a polynomial growth conditon, i.e:
\begin{equation*}
|\Phi(t,y)| \leq C(1+|y|^{k})~~~~\mbox{for some}~k\in\mathbb{N}.
\end{equation*}
Then, for all $(t,\overline{z},y)\in[O,T]\times\mathbb{R}^{+}\times\mathbb{R}^{2}$
\begin{equation*}
\tilde{J}(t,\overline{z},y)\leq-1-\ln(\overline{z})-\Phi(t,y),
\end{equation*}
and for the optimal risk $\tilde{\nu}=0$, we have $\tilde{J}(t,x,y)=-1-ln(x)-\Phi(t,y)$ . 
\end{theorem}

\begin{proof}

Let $\tilde{J}_{\nu}(t,\overline{z},y)=\mathbb{E}\left[\tilde{U}(Z_{T}^{\nu})|Z_{t}^{\nu}=\overline{z},Y_{t}=y\right]$. From (\ref{new-dual}) and $\tilde{U}(\overline{z})=-1-\ln(\overline{z})$, we have the following expression for $\tilde{J}_{\nu}$:
\begin{equation}\label{decomposition-log}
\tilde{J}_{\nu}(t,\overline{z},y)=-1-\ln(\overline{z})+ \mathbb{E}\left[\dfrac{1}{2}\displaystyle\int_{t}^{T}(\psi^{2}(Y_{s})+\nu_{s}^{2})ds\right].
\end{equation}
 
Let $\nu$ be an arbitrary control process, $Y$ the associated process with $Y_{t}=y$ and define the stopping time
 \begin{equation*}
 \theta_{n}:=T\wedge \inf\{s>t: |Y_{s}-y|\geq n\}. 
 \end{equation*}
Now, let $\Phi$ be a $C^{1,2}$ solution to (\ref{loga-PDE-dual}). Then, by It\^{o}s formula, we have:
\begin{align}
\nonumber
\Phi(\theta_{n},Y_{\theta_{n}})&=\Phi(t,y)+\displaystyle\int_{t}^{\theta_{n}}
\left(\frac{\partial \Phi}{\partial_{t}}+ \dfrac{1}{2}Tr(\Sigma\Sigma^{T}D^{2}_{y}\Phi)+\Gamma^{T}D_{y}\Phi\right)(s,Y_{s})ds +\displaystyle\int_{t}^{\theta_{n}}((D_{y}\Phi)^{T} \Sigma)(s,Y_{s})d\overline{W}_{s}\\
&\leq \Phi(t,y)+\dfrac{1}{2}\displaystyle\int_{t}^{\theta_{n}}(\psi^{2}(Y_{s})+\nu_{s}^{2})ds +\displaystyle\int_{t}^{\theta_{n}}((D_{y}\Phi)^{T} \Sigma)(s,Y_{s})d\overline{W}_{s}
\label{inequality-phi}
\end{align}
From the definition of $\theta_{n}$, the integrand in the stochastic integral is bounded on $[t,\theta_{n}]$, a consequence of the continuity of $D_{y}\Phi$ and assumption \textbf{H} $i)$. Then, by taking expectation, one obtains:
\begin{equation*}
\mathbb{E}[\Phi(\theta_{n},Y_{\theta_{n}})]\leq \Phi(t,y) +\mathbb{E}\left[\dfrac{1}{2}\displaystyle\int_{t}^{\theta_{n}}(\psi^{2}(Y_{s})+\nu_{s}^{2})ds\right] .
\end{equation*} 
We now take the limit  as $n$ increases to infinity, then  $\theta_{n}\to T a.s$. From the growth condition satisfied by $\Phi$ and (\ref{integrability-condition}), we can deduce the uniform integrability of $(\Phi(\theta_{n},Y_{\theta_{n}}))_{n}$. Therefore, it follows from the dominated convergence theorem and the boundary condition (\ref{loga-boundary-dual}) that for all $\nu \in \mathcal{K}$:
\begin{equation*}
-\Phi(t,y)\leq \mathbb{E}\left[\dfrac{1}{2}\displaystyle\int_{t}^{T}(\psi^{2}(Y_{s})+\nu_{s}^{2})ds\right] 
\end{equation*}  
Then from  (\ref{log-transf}), we have:
\[\tilde{J}_{\nu}(t,\overline{z},y)\leq -1-\ln(\overline{z})-\Phi(t,y).
\]
Now by repeating the above argument by replacing $\nu$ by $\tilde{\nu}=0$ which is the optimal risk, we can finally deduce that:
\[
\tilde{J}_{\tilde{\nu}}(t,\overline{z},y)=-1-\ln(\overline{z})-\Phi(t,y).
\]
which ends the proof since $\tilde{J}(t,\overline{z},y)=\inf_{\nu\in\mathcal{K}}\tilde{J}_{\nu}(t,\overline{z},y)$
\end{proof}

Let us now study the regularity of the solution $\Phi$ to the semilinear (\ref{loga-PDE-dual}) with the terminal condition (\ref{loga-boundary-dual}).

\begin{proposition}\label{regularity-log}
Under assumptions \textbf{H} $i)$ and $ii)$, there exists a solution $\Phi \in C^{1,2}([0,T)\times \mathbb{R}^{2})\cap C^{0}([0,T]\times \mathbb{R}^{2})$ with polynomial qrowth condition in $y$, to the semilinear (\ref{loga-PDE-dual}) with the terminal condition (\ref{loga-boundary-dual}).  
\end{proposition}
\begin{proof}
Under assumptions $i)$ and $ii)$ and the fact that the Hamiltonian $H$ satisfies a global Lipschitz condition on $D_{y}\Phi$, we can deduce from theorem $4.3$ in Fleming and soner \cite[p.163]{Fleming-soner} the existence and uniqueness of a classical solution to the semilinear equation(\ref{loga-PDE-dual}). 
\end{proof}

\vspace{3mm}
\noindent \subsubsection{\underline{Power utility:}}
As the above reasons given in the logarithmic case, we can suggest that the value function must be of the form:

\begin{equation}\label{power-form-dual}
\tilde{J}(t,\overline{z},y)=-\dfrac{\overline{z}^{q}}{q}\exp(-\Phi(t,y)). 
\end{equation}
Then if we substitute the above form in (\ref{new-HJB-dual})and (\ref{boundary-condition-dual}), we can deduce the following semilinear P.D.E for $\Phi$:
\begin{align}\label{semi-linear-equation-dual}
&-\frac{\partial \Phi}{\partial_{t}}-\frac{1}{2}Tr\left(\Sigma \Sigma^{T} D_{y}^{2}\Phi\right)+H(y,D_{y}\Phi)=0,\\
&~~\Phi(T,y)=0. 
\label{terminal-condition-power-dual}
\end{align}

\vspace{2mm} 
\noindent The Hamiltonian $H$ is defined by:
\begin{align}
\label{first-expression-dual}
H(y,Q)&=\frac{1}{2}Q^{T}\Sigma(y)\Sigma^{T}(y)Q-Q^{T}\Gamma(y) +\displaystyle\inf_{\nu\in\mathcal{K}}\left[\dfrac{1}{2}q(q-1)(\psi^{2}(y)+\nu^{2})+q\left(\psi(y)K_{1}^{T}+\nu K_{2}^{T}\right)Q\right]\\
&= \frac{1}{2}Q^{T}\left(\Sigma(y) \Sigma^{T}(y)-G(y)\right)Q - Q^{T}F(y)+\Psi(y).
\label{second-expression-dual}
\end{align}

\noindent where for $y:=(v,m)$: 
\begin{align*}
&G(y)=\frac{q}{q-1} K_{2}(y) K_{2}^{T}(y)\\
& F(y)= \Gamma(y)-q\psi(y)K_{1}\\
& \Psi(y)=\frac{1}{2}q(q-1)\psi^{2}(y).
\end{align*}

We now state a verification result for the power case, which relates the solution of the above semilinear (\ref{semi-linear-equation-dual}) and (\ref{terminal-condition-power-dual}) to the stochastic control problem (\ref{new-dual}). 

\begin{theorem}[verification theorem]\label{power-verification-dual}
Let assumptions \textbf{H} $i)$, $iii)$ and $iv)$ hold. Suppose that there exists a solution $\Phi \in C^{1,2}([0,T)\times \mathbb{R}^{2})\cap C^{0}([0,T]\times \mathbb{R}^{2})$ with linear growth condition on the derivation $D_{y}\Phi$, to the  semilinear (\ref{semi-linear-equation-dual}) with the terminal condition (\ref{terminal-condition-power-dual}).
Then, for all $(t,x,y)\in[O,T]\times\mathbb{R}^{+}\times\mathbb{R}^{2}$
\begin{itemize}
\item[i)] $\tilde{J}(t,\overline{z},y)\leq-\dfrac{\overline{z}^{q}}{q}\exp(-\Phi(t,y))$. 
\end{itemize}
Now, assume that there exists a minimizer $\tilde{\nu}$ of 
\[
\nu\longrightarrow \dfrac{1}{2}q(q-1)\nu^{2}+q\nu K_{2}(y)^{T}D_{y}\Phi
\]  
such that 
\begin{align*}
-\frac{\partial \Phi}{\partial_{t}}-\frac{1}{2}Tr\left(\Sigma \Sigma^{*} D_{y}^{2}\Phi\right)+ H(y,D_{y}\Phi)=0.
\end{align*}
Then 
\begin{itemize}
\item[ii)] $\tilde{J}(t,\overline{z},y)=-\dfrac{\overline{z}^{q}}{q}\exp(-\Phi(t,y))$.
\end{itemize}
and the associated optimal $\tilde{\nu}$ is given by the Markov control $\{\tilde{\nu_{t}}=\tilde{\nu}(t,Y_{t})\}$ with
\begin{equation}\label{optimal-dual-nu}
\tilde{\nu}_{t}=-\frac{1}{q-1}K_{2}^{T}(Y_{T})D_{y}\Phi(t,Y_{t}).
\end{equation}

\end{theorem}

\begin{proof}

Let us introduce the new probability $\mathbb{Q}^{\nu}$ as follows:
\begin{equation*}
\frac{d\mathbb{Q}^{\nu}}{d\mathbb{P}}=\exp\left(-\displaystyle\int^{t}_{0}q\psi(Y_{u}) d\overline{W}^{1}_{u}-\displaystyle\int^{t}_{0}q \nu_{u} d\overline{W}^{2}_{u}-\frac{1}{2}\displaystyle\int_{0}^{t}q^{2}(\psi^{2}(Y_{u})+\nu_{u}^{2})du\right),
\end{equation*}
From assumption $iv)$ the probability measure$\mathbb{Q}^{\nu}$ with the density process $\frac{d\mathbb{Q}^{\nu}}{d\mathbb{P}}$ is well defined, see Liptser and Shiryaev \cite[P.233]{LiptserShiryaev}.

\vspace{3mm}
\noindent Let $\tilde{J}_{\nu}(t,\overline{z},y)=\mathbb{E}\left[\tilde{U}(Z_{T}^{\nu})|Z_{t}^{\nu}=\overline{z},Y_{t}=y\right]$. 

\noindent From (\ref{new-dual}) and $\tilde{U}(\overline{z})=-\dfrac{\overline{z}^{q}}{q}$, we have from It\^{o}'s formula the following expression for $\tilde{J}_{\nu}$: 
\begin{align}\label{nouvelle-expression-dual}
\tilde{J}_{\nu}(t,\overline{z},y)&=-\frac{\overline{z}^{q}}{q}\mathbb{E}^{\nu}\left[\exp\left(\displaystyle\int_{t}^{T}\dfrac{1}{2}q(q-1)(\psi^{2}(Y_{u})+\nu^{2})du\right)|Y_{t}=y\right].
\end{align}
\noindent Also by Girsanov's theorem, the dynamics of $Y$ under $\mathbb{Q}^{\nu}$, is given by:
\begin{equation}\label{new-dynamic-of-Y}
dY_{t}=\left(\Gamma(Y_{t})-q\psi(Y_{t})K_{1}(Y_{t})-q\nu_{t}K_{2}(Y_{t})\right)dt+\Sigma(Y_{t})dW_{t}^{\nu},
\end{equation} 
where $W^{\nu}$ is a bi-dimensional Brownian motion under $\mathbb{Q}^{\nu}$.

\vspace{2mm}
\noindent Now, let $\Phi$ be a $C^{1,2}$ solution to (\ref{semi-linear-equation-dual}), then by It\^{o}'s formula applied to $\Phi(t,Y_{t})$ under $\mathbb{Q}^{\nu}$, one obtains:
\begin{align*}
\Phi(\theta_{n},Y_{T})=\Phi(t,y)&+ \displaystyle\int_{t}^{T}\left( \frac{\partial \Phi}{\partial_{t}}+\left(\Gamma -q\psi K_{1}-q\nu_{t}K_{2}\right)^{T}D_{y}\Phi +\frac{1}{2}Tr(\Sigma \Sigma^{T}~D^{2}_{y}\Phi) \right)(u,Y_{u})du\\
&+ \displaystyle\int_{t}^{T}(D^{T}_{y}\Phi~~\Sigma)(u,Y_{u})dW_{u}^{\nu}
\end{align*}
Or $\Phi$ is solution of (\ref{semi-linear-equation-dual}), then one obtains:
\begin{align}
\nonumber
\Phi(T,Y_{T})&=\Phi(t,y) + \displaystyle\int_{t}^{T}\left(H(y,D_{y}\Phi) + \left(\Gamma -q\psi K_{1}-q\nu_{t}K_{2}\right)^{T}D_{y}\Phi\right)(u,Y_{u})du\\
&~~~~~~~~~~~~~~+\displaystyle\int_{t}^{T}(D^{T}_{y}\Phi~~\Sigma)(u,Y_{u})dW_{u}^{\nu}\\\nonumber
&\leq \Phi(t,y)+ \displaystyle\int_{t}^{T}\dfrac{1}{2}q(q-1)(\psi^{2}(Y_{u})+\nu^{2})du+ \frac{1}{2}\displaystyle\int_{t}^{T}\left(D^{T}_{y}\Phi~\Sigma \Sigma^{T}D_{y}\Phi\right)(u,Y_{u})du\\
&~~~~~~~~~~~~~~+\displaystyle\int_{t}^{T}(D^{T}_{y}\Phi~~\Sigma)(u,Y_{u})dW_{u}^{\nu},
\label{inequality-withcontrol-dual}
\end{align}
where the inequality comes from the representation (\ref{first-expression-dual}) of the Hamiltonian.

\noindent Therefore, we have: 
\begin{align*}
&\exp(-\Phi(t,y))\mathbb{E}^{\nu}\left[\exp\left(-\frac{1}{2}\displaystyle\int_{t}^{T}\left(D^{T}_{y}\Phi~\Sigma \Sigma^{T}D_{y}\Phi\right)(u,Y_{u})du-\displaystyle\int_{t}^{T}(D^{T}_{y}\Phi~~\Sigma)(u,Y_{u})dW_{u}^{\nu}\right)\right]\\
& \leq \mathbb{E}^{\nu}\left[\displaystyle\int_{t}^{T}\dfrac{1}{2}q(q-1)(\psi^{2}(Y_{u})+\nu^{2})du\right].
\end{align*}
Let us now consider the exponential $Q^{\nu}$-local martingales:
\begin{align*}
\epsilon_{t}^{\pi}=\exp\left(-\displaystyle\int_{0}^{t}(D^{T}_{y}\Phi~~\Sigma)(u,Y_{u})dW_{u}^{\nu}-\frac{1}{2}\displaystyle\int_{0}^{t}\left(D^{T}_{y}\Phi~\Sigma \Sigma^{T}D_{y}\Phi\right)(u,Y_{u})du\right).
\end{align*}
From the Lipschitz condition assumed in $i)$ and from $iii)$, we can deduce from Gronwall's lemma that there exists a positive constant $C$ such that: 
\begin{equation*}
 |Y_{t}|\leq C\left(1+\displaystyle\int_{0}^{t}|W^{\nu}_{u}|du +|W^{\nu}_{t}|\right)
\end{equation*}
Then we deduce that there exists some $\epsilon >0$ such that 
\begin{equation}\label{Grownal}
\sup_{t\in[0,T]}\mathbb{E}^{\nu}[\exp(\epsilon|Y_{t}|^{2})]<\infty.
\end{equation}

Therefore from (\ref{Grownal}) and the fact that $D_{y}\Phi$ satisfies a linear growth condition in y, we can deduce that $\epsilon^{\pi}$ is a martingale under $Q^{\nu}$, therefore we have:

\[
\exp(-\Phi(t,y))\leq \mathbb{E}^{\nu}\left[\displaystyle\int_{t}^{T}\dfrac{1}{2}q(q-1)(\psi^{2}(Y_{u})+\nu^{2})du\right]. 
\] 
The above inequality is proved for all $\nu \in \mathcal{K}$, therefore we can deduce from (\ref{nouvelle-expression-dual}) that:
\[
\tilde{J}(t,\overline{z},y)\leq -\frac{\overline{z}^{q}}{q}\exp(-\Phi(t,y)).
\] 
since $\tilde{J}(t,\overline{z},y)=\inf_{\nu\in\mathcal{K}}\tilde{J}_{\nu}(t,\overline{z},y)$, then $i)$ is proved.

Now by repeating the above argument and observing that the control $\tilde{\nu}$ given by (\ref{optimal-dual-nu}), achieves equality in (\ref{inequality-withcontrol-dual}), we can finally deduce that:
\[
\tilde{J}_{\nu}(t,\overline{z},y)=-\frac{\overline{z}^{q}}{q}\exp(-\Phi(t,y)).
\]
Also since $\tilde{J}(t,\overline{z},y)=\inf_{\nu\in\mathcal{K}}\tilde{J}_{\nu}(t,\overline{z},y)$, then $ii)$ is proved.
\end{proof}

We now study the existence of a classical solution to (\ref{semi-linear-equation-dual})-(\ref{terminal-condition-power-dual}). 

\noindent In fact, the existence of a classical solution to (\ref{semi-linear-equation-dual})-(\ref{terminal-condition-power-dual}) cannot be found directly in the literature since $Q\to H(y,Q)$ is not globally Lipschitz on $Q$ but satisfies a quadratic growth condition on $Q$. For that we can use the approach taken in \cite{Fleming-soner} by considering a certain sequence of approximating P.D.Es which are the HJB-equations of certain stochastic control problems for which the existence of smooth solution is well-known.

\vspace{2mm}
Let us make some assumptions which will be useful to prove the regularity for the solution  of (\ref{semi-linear-equation-dual}). 

\smallskip\smallskip
\underline{\textbf{Assumption (H')}} Let us consider either one of the following conditions:

\smallskip
\underline{$I)$-If $\Sigma$ is a deterministic matrix:} In this case we need the following assumption:

$i)$ $\Gamma$ and $\psi$ are Lipschitz and $C^{1}$ with bounded derivatives. 

\smallskip\smallskip
\underline{$II)$-If $\Sigma$ is not a deterministic matrix:} In this case we need the following assumptions:

\smallskip\smallskip
$i)$ $\Gamma$ and $\psi.K_{1}$ are Lipschitz and $C^{1}$.

$ii)$ $\psi^{2}$, $K_{2}K_{2}^{T}$ are $C^{1}$ with bounded derivatives.

$iii)$ $\Sigma\Sigma^{T}-\dfrac{q}{q-1}K_{2}K_{2}^{T}$ is uniformly elliptic. 

\smallskip
\noindent By the similar arguments used by Pham in \cite{Pham-1} and from the standard verification theorem proved by Fleming and soner [Theorem 3.1 P.163]\cite{Fleming-soner}, we can deduce our regularity result for the case when the Hamiltonian is not globally Lipschitz but satisfies a quadratic growth condition. 

\begin{theorem}\label{regularity-power}
Under one the assumptions (\textbf{H'}), there exists a solution $\Phi \in C^{1,2}([0,T)\times \mathbb{R}^{2})\cap C^{0}([0,T]\times \mathbb{R}^{2})$ with linear growth condition on the derivation $D_{y}\Phi$, to the  semilinear (\ref{semi-linear-equation-dual}) with the terminal condition (\ref{terminal-condition-power-dual}).
\end{theorem}

\begin{remark}
In general, a closed form solution to (\ref{semi-linear-equation-dual}) with the terminal condition (\ref{terminal-condition-power-dual}) does not exist. But we show that for some stochastic volatility model and in the case when the filters estimate are Gaussian, we can obtain a closed form, see section .  
\end{remark}

\smallskip
Let us now describe the relation between the optimal trading strategy and the optimal dual optimiser.
 
\subsection{Solution to the primal problem for special utility functions}\label{explicitform}
We have showed from theorem \ref{optimal-wel}, that the optimal wealth, and by consequence the optimal portfolio, depend on the optimal dual optimiser $\tilde{\nu}$. So we will study this relation in the special case of utility functions studied above. 

From theorem \ref{optimal-wel}, we have:
\begin{equation}\label{rappel-wealth}
\tilde{R}_{t}=\mathbb{E}\left[\frac{Z_{T}^{\tilde{\nu}}}{Z_{t}^{\tilde{\nu}}}I(z_{x}Z_{T}^{\tilde{\nu}})|\mathcal{G}_{t}\right]
\end{equation}
where $ \tilde{\nu}$ is the optimal dual maximizer and $z_{x}$ is the Lagrange multiplier such that $\mathbb{E}\left[Z_{T}^{\tilde{\nu}}I(z_{x}Z_{T}^{\tilde{\nu}})\right]=x$. 

\vspace{3mm}
Before presenting our result concerning the optimal wealth and the optimal portfolio, in order to avoid any confusion, let us describe the dynamics of the wealth $R_{t}$ in terms of the process $Y_{t}:=(V_{t},\overline{\mu}_{t})$ as follows: 
\begin{equation}\label{forme-Y}
dR_{t}=R_{t}\pi_{t}(\psi(Y_{t})\delta(Y_{t})dt +\delta(Y_{t})d\overline{W}^{1}_{t})
\end{equation}
where $\psi(Y_{t})=\overline{\mu}_{t}$ and $\delta(Y_{t})=g(V_{t})$.

\vspace{3mm}
\noindent\underline{\textbf{Logarithmic utility:}} $U(x)=\ln(x)$.
\begin{proposition}
We suppose that the assumptions of theorems \ref{logarithmic-verification-dual} and \ref{regularity-log} hold. Then the optimal wealth process is given by $\tilde{R}_{t}=\dfrac{x}{Z_{t}^{0}}$. Also the optimal portfolio $\tilde{\pi}$ and the primal value function are given by:
\begin{equation}\label{portf-primal}
\tilde{\pi}_{t}=\dfrac{\psi(Y_{t})}{\delta(Y_{t})}:=\dfrac{\overline{\mu}_{t}}{g(V_{t})}~~~~~~\mbox{and}~~J(x)=\ln(x)-\Phi(0,Y_{0}).
\end{equation}
where $\Phi$ is the solution of the semilinear equation (\ref{loga-PDE-dual}) with boundary condition (\ref{loga-boundary-dual}).
\end{proposition}
\begin{proof}
\noindent In this case we have $I(x)=\dfrac{1}{x}$ and from theorem \ref{logarithmic-verification}, the dual optimizer $\tilde{\nu}=0$. Moreover, the Lagrange multiplier $z_{x}=\dfrac{1}{x}$.  Therefore from (\ref{rappel-wealth}), the optimal wealth is given  by 
\begin{equation}\label{wealth-log}
\tilde{R}_{t}=\dfrac{x}{Z_{t}^{0}}. 
\end{equation}
By applying It\^{o}'s formula to (\ref{wealth-log}) and from proposition \ref{innovation-proc}, we obtain that:
\[d\tilde{R}_{t}=\tilde{R}_{t}\psi(Y_{t})d\tilde{W}^{1}_{t}\]
On the other hand, we have from (\ref{forme-Y}) that $d\tilde{R}_{t}=\tilde{R}_{t}\tilde{\pi}_{t}\delta(Y_{t})d\tilde{W}^{1}_{t}$. Therefore comparing these two expressions for $\tilde{R}_{t}$, we obtain that the optimal portfolio $\tilde{\pi}$ is given by
(\ref{portf-primal}).Finally from the definition of the primal value function and (\ref{wealth-log}), we have $J(x)=\ln(x)-\mathbb{E}[\ln(Z_{T}^{0})]=\ln(x)+1+\tilde{J}(0,1,Y_{0})=\ln(x)-\Phi(0,Y_{0})$. The last equality comes from theorem \ref{logarithmic-verification-dual}. 
\end{proof}

\smallskip\smallskip
\noindent\underline{\textbf{Power utility:}} $U(x)=x^{p}/p~~0<p<1$.
\begin{proposition}\label{optimal-port-pow}
We suppose the assumptions of theorems \ref{power-verification-dual} and \ref{regularity-power} hold. Then the optimal wealth is given by:
\[
\tilde{R}_{t}=\dfrac{x}{\mathbb{E}[(Z_{T}^{\tilde{\nu}})^{q}]}(Z_{t}^{\tilde{\nu}})^{q-1}\exp\left(-\Phi(t,Y_{t})\right).
\]
the associated optimal portfolio is given by the Markov control $\{\tilde{\pi}_{t}=\tilde{\pi}(t,Y_{t})\}$ with 
\begin{equation}\label{portf-primal-power}
\tilde{\pi}_{t}=\dfrac{1}{1-p}\dfrac{\psi(Y_{t})}{\delta(Y_{t})}-\dfrac{K_{1}^{T}(Y_{t})}{\delta(Y_{t})}D_{y}\Phi(t,Y_{t})
\end{equation}
and the primal value function is given by:
\[
J(x)=\dfrac{x^{p}}{p}\exp(-(1-p)\Phi(0,Y_{0})).
\]

\noindent Where $q=\dfrac{p}{p-1}$,  $\tilde{\nu}$ is given by (\ref{optimal-dual-nu}) and $\Phi$ is a solution of the semilinear equation (\ref{semi-linear-equation-dual}) with boundary condition (\ref{boundary-condition-dual}).
\end{proposition}

\begin{proof}
In this case we have $I(x)=x^{1/(p-1)}$ and from theorem \ref{power-verification}, the dual optimizer $\tilde{\nu}$ is given by (\ref{optimal-dual-nu}). The Lagrange multiplier $z_{x}=\left(\dfrac{x}{\mathbb{E}[Z_{T}^{p/p-1}]}\right)^{p-1}$.  Therefore from (\ref{rappel-wealth}), the optimal wealth is given by 
\begin{align*}
\tilde{R}_{t}=\mathbb{E}\left[\frac{Z_{T}^{\tilde{\nu}}}{Z_{t}^{\tilde{\nu}}}I(z_{x}Z_{T}^{\tilde{\nu}})|\mathcal{G}_{t}\right]&=
\mathbb{E}\left[\frac{Z_{T}^{\tilde{\nu}}}{Z_{t}^{\tilde{\nu}}}(z_{x})^{1/(p-1)}(Z_{T}^{\tilde{\nu}})^{1/(p-1)}|\mathcal{G}_{t}\right]\\
&=\dfrac{x}{\mathbb{E}[(Z_{T}^{\tilde{\nu}})^{q}]}\dfrac{1}{Z_{t}^{\tilde{\nu}}}\mathbb{E}\left[(Z_{T}^{\tilde{\nu}})^{q}|\mathcal{G}_{t}\right]
\end{align*}
Therefore from theorem \ref{power-verification-dual}, we deduce that:
\begin{equation}\label{wealth-pow}
\tilde{R}_{t}=\dfrac{x}{\mathbb{E}[(Z_{T}^{\tilde{\nu}})^{q}]}(Z_{t}^{\tilde{\nu}})^{q-1}\exp\left(-\Phi(t,Y_{t})\right).
\end{equation}
Now, as in the logarithmic case,  by writing $d\tilde{R}_{t}=\tilde{R}_{t}\pi_{t}\delta(V_{t})d\tilde{W}^{1}_{t}$ and applying It\^{o}'s formula to $(Z_{t}^{\tilde{\nu}})^{q-1}\exp\left(-\Phi(t,Y_{t})\right)$, then  after comparing the two expressions for $\tilde{R}_{t}$, we deduce that:
\[
\tilde{\pi}_{t}=\dfrac{1}{1-p}\dfrac{\psi(Y_{t})}{\delta(Y_{t})}-\dfrac{K_{1}^{T}(Y_{t})}{\delta(Y_{t})}D_{y}\Phi(t,Y_{t}).
\]
Finally, from (\ref{wealth-pow}) and the boundary condition $\Phi(T,Y_{t})=0$, we have:
\[
J(x)=\dfrac{x^{p}}{p}\mathbb{E}[(Z_{T}^{\tilde{\nu}})^{q}]^{1-p}=\dfrac{x^{p}}{p}\exp(-(1-p)\Phi(0,Y_{0})).
\]
where the last equality comes from theorem \ref{power-verification-dual}.
\end{proof}

Let us now deduce the following relation between the primal and dual control function.
\begin{corollary}
The optimal portfolio $\tilde{\pi}$ is given by
\begin{equation}\label{relation}
\tilde{\pi}_{t}=\dfrac{1}{1-p}\dfrac{\psi(Y_{t})}{\delta(Y_{t})}-\dfrac{1}{1-p}\dfrac{K_{1}^{T}(K_{2}^{T})^{-1}}{\delta(Y_{t})}\tilde{\nu}_{t}.
\end{equation}
\end{corollary}
\begin{proof}
The proof can be deduced easily from theorem \ref{power-verification} and proposition \ref{optimal-port-pow}.
\end{proof}
%
\begin{remark}
For the logarithmic case, we notice that in the case of partial information,
the optimal portfolio can be formally derived from the full information case by replacing
the unobservable risk premium $\tilde{\mu}_{t}$ by its estimate $\overline{\mu}_{t}$. But on the other hand, in the power utility function, this property does not hold and the optimal strategy cannot be derived from the full information case by replacing the risk $\tilde{\mu}_{t}$ by its best estimate $\overline{\mu}_{t}$ due to the last additional term which depend on the filter.

\noindent This property corresponds to the so called separation principle. It is proved in Kuwana \cite{Kuwana} that certainty equivalence holds if and only if the utilities functions are logarithmic.
\end{remark}

\begin{remark}\label{pas-PDE-approach}
The advantage of using the martingale approach instead of the dynamic programming approach (PDE approach) is that we don't need to impose any constraint on the admissible portfolio controls, while it is essential in the case of the PDE approach. In fact, with the PDE approach, we need to make the following constraint on the admissible portfolio controls:
\begin{equation}\label{admissible-condition}
\sup_{t\in[0,T]}\mathbb{E}[\exp(c|\delta(Y_{t})\pi_{t}|)] <\infty,~~~~~~~\mbox{for some}~c>0.
\end{equation} 
this constraint is indispensable to impose in order to show a verification theorem in the case of power utility function. 
\end{remark}

\subsection{Application}
Here we give an example of stochastic volatility model for which we can obtain a closed form for the value function and the optimal portfolio. Let us consider the Log Ornstein-Uhlenbeck model defined in (\ref{model-S-app}), (\ref{model-V-app}) and (\ref{model-mu-app}). Also we consider the power utility function $U(x)=\dfrac{x^{p}}{p},~~0<p<1$. 
 
\smallskip
\noindent Firstly, notice that we have the following dynamics of $(R_{t}^{\pi},V_{t},\overline{\mu}_{t})$ in the full observation framework: 
\begin{align*}
&dR_{t}^{\pi}=\R_{t}^{\pi}\pi_{t}\left(\overline{\mu}_{t}e^{V_{t}}dt+e^{V_{t}}d\overline{W}^{1}_{t}\right)\\
& dV_{t}=\lambda_{V}\left(\theta- V_{t}\right) dt + \sigma_{V}\rho d\overline{W}^{1}_{t}+\sigma_{V}\sqrt{1-\rho^{2}} d\overline{W}^{2}_{t}\\
&d\overline{\mu}_{t}=\left(-\lambda_{\mu}\overline{\mu}_{t}+\lambda_{\mu}\theta_{\mu}\right)dt +\Theta_{t}^{11}d\overline{W}^{1}_{t}+\Theta_{t}^{12}d\overline{W}^{2}_{t}.
\end{align*}
 where the last dynamics is deduced from (\ref{filter-estimates-application}). $\Theta^{11}$ and $\Theta^{12}$ are solutions of Riccati equation (\ref{Riccati-equation}). 
 
Therefore the primal value function $J(x)$ and the associated optimal portfolio $\tilde{\pi}_{t}$ are given explicitly.  

\begin{proposition}\label{explicit-form-exmaple}
The optimal portfolio is given by: 
\[
\tilde{\pi}_{t}=\dfrac{1}{p-1}\dfrac{\overline{\mu}_{t}}{e^{V_{t}}}-\dfrac{\rho\sigma_{V}}{e^{V_{t}}}[\tilde{A}(t)+(T-t)] +\dfrac{\Theta_{11}}{e^{V_{t}}}(2\overline{A}(t)\overline{\mu}_{t}+\overline{B}(t)).
\]
and the primal value function is given by:
\[
J(x)=\dfrac{x^{p}}{p}\exp\big[-(1-p)\big(\tilde{A}(0)V_{0}+\tilde{B}(0)-V_{0}T-\overline{A}(0)\overline{\mu}^{2}_{0}-\overline{B}(0)\overline{\mu}_{0}-\overline{C}(0)\big)\big].
\]
where:
 \begin{align*}
 &\tilde{A}(t)=-\lambda_{V}\displaystyle\int_{t}^{T}(T-s)e^{-\lambda_{V}(s-t)}ds. \\
 &\tilde{B}(t)=\displaystyle\int_{t}^{T}\Big[-\dfrac{1}{2}(\sigma_{V}^{2}-\dfrac{q}{q-1}(1-\rho^{2})\sigma_{V}^{2})A^{2}(s)+\left((\sigma_{V}^{2}-\dfrac{q}{q-1}(1-\rho^{2})\sigma_{V}^{2})+\lambda_{V}\theta\right)A(s)\\
&~~~~~~~~~~-\dfrac{1}{2}(\sigma_{V}^{2}-\dfrac{q}{q-1}(1-\rho^{2})\sigma_{V}^{2})(T-s)^{2}-\lambda_{V}\theta(T-s)\Big] ds.
\end{align*}
and $\overline{A}$ is solution of the following Riccati equation:
\begin{align*}
\overline{A}^{'}(t)=-2\left(\Theta_{11}^{2}+\Theta_{12}^{2}-\dfrac{q}{q-1}\Theta_{12}^{2}\right)\overline{A}^{2}(t)+2(\lambda_{\mu}+q\Theta_{11})\overline{A}(t)-\dfrac{1}{2}q(q-1),~~\mbox{with}~~\overline{A}(T)=0
\end{align*}
and 
\begin{align*}
&\overline{B}(t)=\displaystyle\int_{t}^{T}B_{1}(s)\overline{A}(s) \exp\big[-(\lambda_{\mu}+q\Theta_{11})(s-t)+2(\Theta_{11}^{2}+\Theta_{12}^{2}-\dfrac{q}{q-1}\Theta_{12}^{2})\displaystyle\int_{t}^{s}\overline{A}(u)du\Big]ds.\\
&\overline{C}(t)=\displaystyle\int_{t}^{T}\Big[(\Theta_{11}^{2}+\Theta_{12}^{2})\overline{A}(s)+\dfrac{1}{2}\left(\Theta_{11}^{2}+\Theta_{12}^{2}-\dfrac{q}{q-1}\Theta_{12}^{2}\right)\overline{B}^{2}(s)-B_{1}(s)\overline{B}(s)\Big]ds.
 \end{align*}
where
 \[
B_{1}(s)=+2\left[(\rho\sigma_{V}\Theta_{11}+\sqrt{1-\rho^{2}}\sigma_{V}\Theta_{12}-\dfrac{q}{q-1}\sqrt{1-\rho^{2}}\sigma_{V}\Theta_{12})(\tilde{A}(s)-(T-s))-\lambda_{\mu} \theta_{\mu}\right]
 \]
and with terminal conditions: $\tilde{A}(T)=\tilde{B}(T)=\overline{A}(T)=\overline{B}(T)=\overline{C}(T)=0$.
\end{proposition}
\begin{proof}
See Appendix A. 
\end{proof}

\newpage
\appendix
\section{Appendix}
\underline{\textbf{Filtering}}
Let us consider the following partially observation system: 
\begin{align}
\label{signal-process}
&dX_{t}=A(X_{t})dt +G(X_{t})dM_{t}+ B(X_{t})dW_{t}\\
&dY_{t} = dW_{t}+h(X_{t}) dt
\label{observation-process}
\end{align}
Here $X$ is the two dimensional signal process and $Y$ is the two dimensional observation process. $A$ is a $2\times 1$ matrix, $G, B$ are $2\times 2$ matrix and $h$ is $2\times 1$ matrix . $W$ and $M$ are two dimensional independents Brownian motions. 
 
Now, we will be interested in the filtering problem which consists in evaluating the conditional expectation of the unobservable process having the observations. In the sequel, we denote this conditional expectation by $\alpha_{t}(\phi)=\mathbb{E}\left[\phi(X_{t})|\mathcal{F}_{t}^{Y}\right]$, where $\mathbb{F}^{Y}$ is the filtration generated by the observation process $Y$.

\vspace{2mm}

\noindent Then one of the approaches to obtain the evolution equation for $\alpha_{t}$ is to change
the measure. Using the change of measure $\tilde{\mathbb{P}}$ given in (\ref{equivalent-pro}), we can define a new measure $\tilde{\mathbb{P}}$, such that the observation process becomes a $\tilde{\mathbb{P}}$ Brownian motion independent of the signal variable $X_{t}$. For that we need to discuss some conditions under which the process $L$ is a martingale:
 \begin{equation}\label{martingale-L}
 L_{t}=\exp\left(-\sum_{i=1}^{2}\displaystyle\int_{0}^{t}h_{i}(X_{s})dW^{i}_{s}-\frac{1}{2}\sum_{i=1}^{2}\displaystyle\int_{0}^{t}h_{i}(X_{s})^{2}ds\right).
 \end{equation}

\noindent Firstly, the classical condition is Novikov's condition:
\begin{align*}
\mathbb{E}\left[\exp\left(\frac{1}{2}\displaystyle\int_{0}^{t}h_{1}(X_{s})^{2}ds+\frac{1}{2}\displaystyle\int_{0}^{t}h_{2}(X_{s})^{2}ds\right)\right]<\infty.
\end{align*}
Normally Novikov's condition is quite difficult to verify directly, so we need to use an alternative conditions under which the process $L$ is a martingale. 

\noindent From lemma $3.9$ in \cite{Bain}, we can deduce that $L$ is a martingale if the following conditions are satisfied:
\begin{align}\label{assumption-L}
\mathbb{E}\left[\displaystyle\int_{0}^{t}(||h(X_{s})||^{2}) ds\right]<\infty,~~~~~~~\mathbb{E}\left[\displaystyle\int_{0}^{t}L_{s}||h(X_{s})||^{2}ds\right]<\infty~~~~~~~\forall t>0.
\end{align} 
Let us now denote by $\Lambda_{t}$ the $\left(\tilde{\mathbb{P}},\mathbb{F}\right)$-martingale given 
by $\Lambda_{t}=\frac{1}{L_{t}}$. We then have:
\begin{align*}
\frac{d\mathbb{P}}{d\tilde{\mathbb{P}}}|\mathcal{F}_{t}&=\Lambda_{t},~~~~0 \leq t \leq~T\\
&=\exp\left(\sum_{i=1}^{2}\displaystyle\int_{0}^{t}h_{i}(X_{s})dW^{i}_{s}-\frac{1}{2}\sum_{i=1}^{2}\displaystyle\int_{0}^{t}h_{i}(X_{s})^{2}ds\right).
\end{align*}

\vspace{3mm}
\noindent Therefore the computation of $\alpha_{t}(\phi)$ is obtained by the so-called Kallianpur-Striebel formula, which is related to Bayes formula. For every  $\phi \in \mathbb{B}(\mathbb{R}^{d})$, we have the following representation:

\begin{equation}\label{Kallianpur}
\alpha_{t}(\phi):=\mathbb{E}\left[\phi(X_{t})|\mathcal{F}^{Y}_{t}\right]=\frac{\tilde{\mathbb{E}}\left[\phi(X_{t})\Lambda_{t}|\mathcal{G}^{Y}_{t}\right]}{\tilde{\mathbb{E}}\left[\Lambda_{t}|\mathcal{G}^{Y}_{t}\right]}:=\frac{\psi_{t}(\phi)}{\psi_{t}(1)},
\end{equation}

\noindent with $\psi_{t}(\phi):=\tilde{\mathbb{E}}[\phi(X_{t})\Lambda_{t}|\mathcal{G}^{Y}_{t}]$ is the unnormalized conditional distribution of $\phi(X_{t})$, given $\mathcal{G}^{Y}_{t}$, $\psi_{t}(1)$ can be viewed as the normalising factor and $\mathbb{B}(\mathbb{R}^{d})$ is the space of bounded measurable functions $\mathbb{R}^{2}\to\mathbb{R}$. 

\vspace{2mm}
\noindent In the following, we assume that for all $t\geq 0$,
 \begin{align}
\tilde{\mathbb{P}}\left[\displaystyle \int_{0}^{t}[\psi_{s}(||h||)]^{2}ds <\infty\right]=1, ~~\mbox{for all}~t >0.
\label{cond-zakai-1}
\end{align}

\noindent Let us now introduce the following notations which will be useful in the sequel. 

\begin{notations}
\noindent Let $K=\dfrac{1}{2}(BB^{T}+GG^{T})$ and $\mathcal{A}$ be the generator associated with the process $X$ in the second order differential operator:
\begin{equation}\label{Operator-A}
\mathcal{A}\phi=\sum_{i,j=1}^{2}K_{ij}\partial_{x_{i}x_{j}}^{2}\phi +\sum_{i=1}^{2}A_{i}\partial_{x_{i}}\phi, ~~~~~~~~\mbox{for}~\phi\in \mathbb{B}(\mathbb{R}^{d}).
\end{equation}
and its adjoint $\mathcal{A}^{*}$ is given by:
\begin{equation}\label{adjoint-A}
\mathcal{A}^{*}\phi=\sum_{i=1}^{2}\partial_{x_{i}x_{j}}^{2}(K_{ij}\phi)-\sum_{i=1}^{2}\partial_{x_{i}}(A_{i}\phi).
\end{equation}
Also we introduce the following operator $\mathcal{B}=(\mathcal{B}^{k})_{k=1}^{2}$:
\begin{equation}\label{operator-B}
\mathcal{B}^{k} \phi=\sum_{i=1}^{2}B_{ik}\partial_{x_{i}}\phi, ~~~~~~~~\mbox{for}~\phi\in \mathbb{B}(\mathbb{R}^{d}).
\end{equation}
and the adjoint of the operator $\mathcal{B}$ is given by $\mathcal{B}^{k,*}=(\mathcal{B}^{k,*})_{k=1}^{2}$:
\begin{equation}\label{adjoint-operator-B}
\mathcal{B}^{1,*}\phi=-\sum_{i=1}^{2}\partial_{x_{i}}(B_{i1}\phi),~~~~\mathcal{B}^{2,*}\phi=-\sum_{i=1}^{2}\partial_{x_{i}}(B_{i2}\phi).
\end{equation} 
\end{notations}
\noindent The following two propositions show that the unnormalized conditional distribution (resp. the conditional distribution) of the signal is a solution of a linear stochastic partial differential equation often called the Zakai equation (resp. nonlinear stochastic and parabolic type partial differential equation often called the Kushner-Stratonovich equation). These results due to Bain and Crisan \cite{Bain} and Pardoux \cite{pardoux}.


\begin{proposition}\label{general-Zakai-equation}
Assume that the signal and observation processes satisfy (\ref{signal-process}) and (\ref{observation-process}). If conditions (\ref{assumption-L}) and (\ref{cond-zakai-1}) are satisfied then the unnormalized conditional distribution $\psi_{t}$ satisfies the following Zakai equation:
 \begin{equation}\label{general-zakai-ex}
d\psi_{t}(\phi)=\psi_{t}(A\phi)dt+\psi_{t}\left(\left(h^{1}+\mathcal{B}^{1}\right)\phi\right)d\tilde{W}_{t}^{1}+\psi_{t}\left(\left(h^{2}+\mathcal{B}^{2}\right)\phi\right)d\tilde{W}_{t}^{2}.
 \end{equation}
   $ \mbox{for any}~~\phi\in B(\mathbb{R}^{2})$.

\end{proposition}
\begin{proposition}\label{general-Kushner-equation}
Assume that the signal and observation processes satisfy (\ref{signal-process}) and (\ref{observation-process}). If conditions (\ref{assumption-L}) and (\ref{cond-zakai-1}) are satisfied then the conditional distribution $\alpha_{t}$ satisfies the following Kushner-Stratonovich equation:
 \begin{align}
 \nonumber
d\alpha_{t}(\phi)&=\alpha_{t}(A\phi)dt +\left[\alpha_{t}\left(\left(h^{1}+\mathcal{B}^{1}\right)\phi\right)-\alpha_{t}(h^{1})\alpha_{t}(\phi)\right]d\overline{W}_{t}^{1}\\
&~~~~~~~~~~~~~~~~+\left[\alpha_{t}\left(\left(h^{2}+\mathcal{B}^{2}\right)\phi\right)-\alpha_{t}(h^{2})\alpha_{t}(\phi)\right]d\overline{W}_{t}^{2}.
\label{general-kushner-ex}
 \end{align}
 $ \mbox{for any}~~\phi\in B(\mathbb{R}^{2})$.
\end{proposition}

\noindent \textbf{Proof of lemma \ref{EMM-opt}}

From equation (\ref{exp-EMM}), the definition of the conditional expectation and Jensen's inequality, it follows for any $\nu \in \mathcal{K}$:
\begin{align*}
&\mathbb{E}\left[\tilde{U}\left(z Z_{T}^{\nu}\right)\right]\\
&=\mathbb{E}\left[\mathbb{E}\left[\tilde{U}\left(z \exp\left(-\displaystyle\int_{0}^{T}\overline{\mu}_{s}d\overline{W}^{1}_{s}
-\frac{1}{2}\displaystyle\int_{0}^{T}\overline{\mu}_{s}^{2} ds -
\displaystyle\int_{0}^{T}\nu_{s}d\overline{W}^{2}_{s}
-\frac{1}{2}\displaystyle\int_{0}^{T}\nu_{s}^{2} ds\right)\right)|\mathcal{F}_{T}^{\tilde{W}^{1}}\right]\right]\\
&\geq\mathbb{E}\left[\tilde{U}\left(z \exp\left(-\displaystyle\int_{0}^{T}\overline{\mu}_{s}d\overline{W}^{1}_{s}
-\frac{1}{2}\displaystyle\int_{0}^{T}\overline{\mu}_{s}^{2} ds\right)\mathbb{E}\left[\exp\left(-
\displaystyle\int_{0}^{T}\nu_{s}d\overline{W}^{2}_{s}
-\frac{1}{2}\displaystyle\int_{0}^{T}\nu_{s}^{2} ds\right)|\mathcal{F}_{T}^{\tilde{W}^{1}}\right]\right)\right].
\end{align*}

\noindent On the other hand,  $\mathbb{E}\left[\exp\left(-
\displaystyle\int_{0}^{T}\nu_{s}d\overline{W}^{2}_{s}
-\frac{1}{2}\displaystyle\int_{0}^{T}\nu_{s}^{2} ds\right)\right]=1$ a.s. In fact, from the definition of the conditional expectation, it remains to prove that for each positive function $h$, for each, $t_{1},......t_{k} \in [0,T]$, we have:

\begin{equation*}
\mathbb{E}\left[\exp\left(-
\displaystyle\int_{0}^{T}\nu_{s}d\overline{W}^{2}_{s}
-\frac{1}{2}\displaystyle\int_{0}^{T}\nu_{s}^{2} ds\right)h\left(\tilde{W}^{1}_{t_{1}},.....\tilde{W}^{1}_{t_{k}}\right)\right]=\mathbb{E}\left[h\left(\tilde{W}^{1}_{t_{1}},.....\tilde{W}^{1}_{t_{k}}\right)\right].
\end{equation*}

\noindent As $\nu$ is a $\mathbb{G}$-adapted, we can define a new probability measure $\mathbb{P}^{\nu}$ equivalent to $\mathbb{P}$ on $\mathcal{G}_{T}$ given by:
\begin{equation*}
\frac{d\mathbb{P}^{\nu}}{d \mathbb{P}}=\exp\left(-\displaystyle\int_{0}^{T}\nu_{u}d\overline{W}^{2}-\frac{1}{2}\displaystyle\int_{0}^{T}\nu_{u}^{2}du\right)
\end{equation*}

\noindent By Girsanov theorem, $N$ is a $\mathbb{G}$ Brownian motion under $\mathbb{P}^{\nu}$. On the other hand, from the dynamic of $\tilde{W}^{1}$ given by $d\tilde{W}^{1}=dN_{t}+\overline{\mu}_{t}dt$ and the assumption that $\overline{\mu}_{t}\in\mathcal{F}^{\tilde{W}^{1}}_{t}$, we deduce that the law of $\tilde{W}^{1}$ remains the same under $\mathbb{P}$ and $\mathbb{P}^{\nu}$. Thus:

\begin{align*}
\mathbb{E}^{\nu}\left[h\left(\tilde{W}^{1}_{t_{1}},.....\tilde{W}^{1}_{t_{k}}\right)\right]&:=
\mathbb{E}\left[\exp\left(-
\displaystyle\int_{0}^{T}\nu_{s}d\overline{W}^{2}_{s}
-\frac{1}{2}\displaystyle\int_{0}^{T}\nu_{s}^{2} ds\right)h\left(\tilde{W}^{1}_{t_{1}},.....\tilde{W}^{1}_{t_{k}}\right)\right]\\
&=\mathbb{E}\left[h\left(\tilde{W}^{1}_{t_{1}},.....\tilde{W}^{1}_{t_{k}}\right)\right].
\end{align*} 

\noindent Therefore $\mathbb{E}\left[\exp\left(-
\displaystyle\int_{0}^{T}\nu_{s}d\overline{W}^{2}_{s}
-\frac{1}{2}\displaystyle\int_{0}^{T}\nu_{s}^{2} ds\right)\right]=1$ and then one obtains:
\begin{equation*}
\mathbb{E}\left[\tilde{U}\left(z Z_{T}^{\nu}\right)\right]\geq 
\mathbb{E}\left[\tilde{U}\left( z \exp\left(-\displaystyle\int_{0}^{t}\overline{\mu}_{s}d\overline{W}^{1}_{s}
-\frac{1}{2}\displaystyle\int_{0}^{t}\overline{\mu}_{s}^{2} ds\right)\right)\right]:=\mathbb{E}\left[\tilde{U}\left(z Z_{T}^{0}\right)\right].
\end{equation*}

\noindent On the other hand, we have from the definition of the dual problem that $\tilde{J}(z)\leq\mathbb{E}\left[\tilde{U}\left(z Z_{T}^{0}\right)\right] $, so we conclude that 
\begin{equation*}
J_{dual}(z)=\mathbb{E}\left[\tilde{U}\left(z Z_{T}^{0}\right)\right].
\end{equation*}

\vspace{5mm}
\textbf{Proof of proposition \ref{explicit-form-exmaple}}

With the Log-Ornstein model given by (\ref{model-S-app}),(\ref{model-V-app}) and (\ref{model-mu-app}), the assumptions \textbf{H} and \textbf{H'} $i)$ hold. Therefore from proposition \ref{optimal-port-pow}, we have
\[
\tilde{\pi}_{t}=\dfrac{1}{p-1}\dfrac{\overline{\mu}_{t}}{e^{V_{t}}}-\dfrac{K_{1}^{T}(Y_{t})}{e^{V_{t}}}D_{y}\Phi(t,Y_{t}),
\] 
where $K_{1}^{T}=(\rho\sigma_{V}~~~\Theta_{11})$ and $\Phi$ is solution of (\ref{semi-linear-equation-dual}). Generally, equation (\ref{semi-linear-equation-dual}) does not have closed-form, but with this model we can deduce a closed form for $\Phi$ by using the following separation transformation: For $y=(v,m)$, 
\[\Phi(t,y)=\tilde{\Phi}(t,v)-\tilde{f}(t,v,m).\]

The general idea of this separation transformation has been used by a lot of authors like Fleming \cite{Fleming} Pham \cite{Pham-1}, Rishel\cite{Rishel}.., in order to express the value function in terms of the solution to a semilinear parabolic equation. 

\noindent Now, substituting the above form of $\Phi$ into (\ref{semi-linear-equation-dual}) gives us:
\begin{align*}
&-\dfrac{\partial\tilde{\Phi}}{\partial_{t}}+\dfrac{\partial\tilde{f}}{\partial_{t}}-\dfrac{1}{2}\sigma^{2}_{V}\left[\dfrac{\partial^{2}\tilde{\Phi}}{\partial^{2}_{v}}-\dfrac{\partial^{2}\tilde{f}}{\partial^{2}v}\right]+\left(\rho\sigma_{V} \Theta_{11}+\sqrt{1-\rho^{2}}\sigma_{V}\Theta_{12}\right)\dfrac{\partial \tilde{f}}{\partial_{v,m}}+\dfrac{1}{2}(\Theta_{11}^{2}+\Theta_{12}^{2})\dfrac{\partial^{2}\tilde{f}}{\partial^{2}m}\\
&+\dfrac{1}{2}(\sigma_{V}^{2}-\dfrac{q}{q-1}(1-\rho^{2})\sigma_{V}^{2})\left[(\dfrac{\partial\tilde{\Phi}}{\partial_{v}})^{2}-2\dfrac{\partial\tilde{\Phi}}{\partial_{v}}\dfrac{\partial\tilde{f}}{\partial_{v}}+(\dfrac{\partial\tilde{f}}{\partial_{v}})^{2}\right]-\left(\lambda_{V}(\theta-v)-qm\rho\sigma_{V}\right)\left(\dfrac{\partial\tilde{\Phi}}{\partial_{v}}-\dfrac{\partial\tilde{f}}{\partial_{v}}\right)\\
&+\left(-\lambda_{\mu} m +\lambda_{\mu} \theta_{\mu}-qm\Theta_{11}\right)\dfrac{\partial\tilde{f}}{\partial_{m}}+\dfrac{1}{2}\left(\Theta_{11}^{2}+\Theta_{12}^{2}-\dfrac{q}{q-1}\Theta_{12}^{2}\right)(\dfrac{\partial \tilde{f}}{\partial_{m}})^{2}
+\dfrac{1}{2}q(q-1)m^{2}\\
 &-\left(\rho\sigma_{V}\Theta_{11}+\sqrt{1-\rho^{2}}\sigma_{V}\Theta_{12}-\dfrac{q}{q-1}\sqrt{1-\rho^{2}}\sigma_{V}\Theta_{12}\right)(\dfrac{\partial\tilde{f}}{\partial_{m}}\dfrac{\partial\tilde{\Phi}}{\partial_{v}}-\dfrac{\partial\tilde{f}}{\partial_{m}}\dfrac{\partial\tilde{f}}{\partial_{v}})=0.
\end{align*}
Thus we have a coupled PDEs for which we have not able to find its solution in general. The key is to separate the considered PDE into a PDE in $\tilde{\Phi}$  and another in $\tilde{f}$, with the fact that $\tilde{\Phi}(T,v)=0$ and $\tilde{f}(T,v,m)=0$. These two last conditions come from the boundary condition (\ref{boundary-condition-dual}). 

\vspace{3mm}
But, there is also another difficult to obtain a explicit solution for $\tilde{f}$. This difficulty comes from the terms $\dfrac{\partial\tilde{f}}{\partial_{v,m}}$ and $\dfrac{\partial\tilde{f}}{\partial v}\dfrac{\partial\tilde{f}}{\partial m}$. For that we need to impose the following separation form on $\tilde{f}$: $\tilde{f}(t,v,m)=v.(T-t)+ \overline{f}(t,m)$, with $\overline{f}(T,m)=0$.

\noindent Finally, we have the following PDEs for $\tilde{\Phi}$ and $\overline{f}$ for which we can deduce an explicit form as follows: 
\begin{align}
\nonumber
&-\dfrac{\partial\tilde{\Phi}}{\partial_{t}}-\dfrac{1}{2}\sigma^{2}_{V}\dfrac{\partial^{2}\tilde{\Phi}}{\partial^{2}_{v}}+\dfrac{1}{2}(\sigma_{V}^{2}-\dfrac{q}{q-1}(1-\rho^{2})\sigma_{V}^{2})(\dfrac{\partial\tilde{\Phi}}{\partial_{v}})^{2}-\left((\sigma_{V}^{2}-\dfrac{q}{q-1}(1-\rho^{2})\sigma_{V}^{2})+\lambda_{V}(\theta-v)\right)\dfrac{\partial\tilde{\Phi}}{\partial_{v}}\\
&~~~~-\lambda_{V}(T-t)v+\dfrac{1}{2}(\sigma_{V}^{2}-\dfrac{q}{q-1}(1-\rho^{2})\sigma_{V}^{2})(T-t)^{2}+\lambda_{V}\theta(T-t).
\label{equatildephi}
\end{align}
and 
\begin{align}
\nonumber
&\dfrac{\partial\overline{f}}{\partial_{t}}+\dfrac{1}{2}(\Theta_{11}^{2}+\Theta_{12}^{2})\dfrac{\partial^{2}\overline{f}}{\partial^{2}m}+\dfrac{1}{2}\left(\Theta_{11}^{2}+\Theta_{12}^{2}-\dfrac{q}{q-1}\Theta_{12}^{2}\right)(\dfrac{\partial \tilde{f}}{\partial_{m}})^{2}\\\nonumber
&+\left[-(\rho\sigma_{V}\Theta_{11}+\sqrt{1-\rho^{2}}\sigma_{V}\Theta_{12}-\dfrac{q}{q-1}\sqrt{1-\rho^{2}}\sigma_{V}\Theta_{12})(\dfrac{\partial\tilde{\Phi}}{\partial_{v}}-(T-t))-\lambda_{\mu} m +\lambda_{\mu} \theta_{\mu}-qm\Theta_{11}\right]\dfrac{\partial\overline{f}}{\partial_{m}}\\
&+\dfrac{1}{2}q(q-1)m^{2}+qm\rho\sigma_{V}\dfrac{\partial\tilde{\Phi}}{\partial v}-qm\rho\sigma_{V}(T-t)
\label{equatildef}
\end{align}
Notice that the PDE for $\overline{f}$ depends on $\dfrac{\partial\tilde{\Phi}}{\partial_{v}}$, but we show below that the solution of the PDF satisfied by $\tilde{\Phi}$ is polynomial of degree $1$, then by deriving it, we obtain a term which does not depend on $v$. So we have a PDE for $\overline{f}$ which depends only on $m$, therefore an explicit form can be deduced.  

\vspace{3mm}
\noindent The solution of (\ref{equatildephi}) with the boundary condition $\tilde{\Phi}(T,v)=0$ is given by: 
\begin{align*}
\tilde{\Phi}(t,v)=\tilde{A}(t)v+\tilde{B}(t)
\end{align*}
where: $\tilde{A}_{t}$ and $\tilde{B}(y)$ are respectively solutions of the following differential equations:
\begin{align*}
&\tilde{A}^{'}(t)=\lambda_{V}\tilde{A}(t)-\lambda_{V}(T-t)),~~~~~~~~~\mbox{with}~~\tilde{A}(T)=0,\\
&\tilde{B}^{'}(t)=\dfrac{1}{2}(\sigma_{V}^{2}-\dfrac{q}{q-1}(1-\rho^{2})\sigma_{V}^{2})A^{2}(t)-\left((\sigma_{V}^{2}-\dfrac{q}{q-1}(1-\rho^{2})\sigma_{V}^{2})+\lambda_{V}\theta\right)A(t)\\
&~~~~~~~~~~+\dfrac{1}{2}(\sigma_{V}^{2}-\dfrac{q}{q-1}(1-\rho^{2})\sigma_{V}^{2})(T-t)^{2}+\lambda_{V}\theta(T-t)~~~~~\mbox{with}~~\tilde{B}(T)=0,
\end{align*} 

One easily verifies that $\tilde{A}(t)$, $\tilde{B}(t)$ given in proposition \ref{explicit-form-exmaple} are solutions of the above differential equations.

\noindent On the other hand, the solution of (\ref{equatildef}) with the boundary condition  $\overline{f}(T,m)=0$ is given by:   
\begin{align*}
\overline{f}(t,m)=\overline{A}(t)m^{2}+\overline{B}(t)m+\overline{C}(t)
\end{align*}
Where:
\begin{align*}
&\overline{A}^{'}(t)=-2\left(\Theta_{11}^{2}+\Theta_{12}^{2}-\dfrac{q}{q-1}\Theta_{12}^{2}\right)\overline{A}^{2}(t)+2(\lambda_{\mu}+q\Theta_{11})\overline{A}(t)-\dfrac{1}{2}q(q-1),\\
&\overline{B}^{'}(t)=\left[-2(\Theta_{11}^{2}+\Theta_{12}^{2}-\dfrac{q}{q-1}\Theta_{12}^{2})\overline{A}(t)+(\lambda_{\mu}+q\Theta_{11})\right]\overline{B}(t)-q\rho\sigma_{V}\tilde{A}(t)+q\rho\sigma_{V}(T-t)\\
&+2\left[(\rho\sigma_{V}\Theta_{11}+\sqrt{1-\rho^{2}}\sigma_{V}\Theta_{12}-\dfrac{q}{q-1}\sqrt{1-\rho^{2}}\sigma_{V}\Theta_{12})(\tilde{A}(t)-(T-t))-\lambda_{\mu} \theta_{\mu}\right]\overline{A}(t),\\
&\overline{C}^{'}(t)=-(\Theta_{11}^{2}+\Theta_{12}^{2})\overline{A}(t)-\dfrac{1}{2}\left(\Theta_{11}^{2}+\Theta_{12}^{2}-\dfrac{q}{q-1}\Theta_{12}^{2}\right)\overline{B}^{2}(t)\\
&+\left[(\rho\sigma_{V}\Theta_{11}+\sqrt{1-\rho^{2}}\sigma_{V}\Theta_{12}-\dfrac{q}{q-1}\sqrt{1-\rho^{2}}\sigma_{V}\Theta_{12})(\tilde{A}(t)-(T-t))-\lambda_{\mu} \theta_{\mu}\right]\overline{B}(t).
\end{align*}
 with terminal condition $\overline{A}(T)=\overline{B}(T)=\overline{C}(T)=0$.
The solution of the riccati equation satisfied by $\overline{A}(t)$ can be deduced from \cite{Rishel}. For $\overline{B}(t)$ and $\overline{C}(t)$, on easily verifies that their expressions given in proposition \ref{explicit-form-exmaple} are solutions of the above differential equations.

Finally, from proposition \ref{optimal-port-pow} and the above solutions of $\tilde{\Phi}$ and $\overline{f}$, we can deduce the explicit form of the value function given in proposition \ref{explicit-form-exmaple}.

\def\refname{References}
\bibliographystyle{plain}
\bibliography{ref}

\end{document}